\documentclass[twoside,english,12pt]{amsart}
\usepackage{babel}
\usepackage[T1]{fontenc}
\usepackage{lmodern}
\usepackage{geometry} 
\usepackage{amssymb,latexsym,amsmath,amsthm} 
\usepackage{color}
\usepackage{enumerate}

\numberwithin{equation}{section}

\newtheorem{definition}{Definition}[subsection]
\newtheorem{proposition}{Proposition}[subsection]
\newtheorem{theorem}{Theorem}[subsection]
\newtheorem{corollary}{Corollary}[subsection]
\newtheorem{lemma}{Lemma}[subsection]

\def\bA{\mathbf{A}}
\def\bC{\mathbf{C}}

\def\bF{\mathbf{F}}
\def\bG{\mathbf{G}}

\def\bP{\mathbf{P}}
\def\bR{\mathbf{R}}

\def\bZ{\mathbf{Z}}

\def\cF{\mathcal{F}}

\def\R{\mathbb{R}}
\def\C{\mathbb{C}}
\def\fC{\frak{C}} 
\def\e{\varepsilon}

\def\bG{\mathbf{G}}

\def\N{\mathbf{N}}

\def\b1{\mathbf{1}}

\def\cA{\mathcal{A}}
\def\cC{\mathcal{C}}
\def\cD{\mathcal{D}}
\def\cE{\mathcal{E}}
\def\cV{\mathcal{V}}
\def\cH{\mathcal{H}}
\def\cF{\mathcal{F}}
\def\cL{\mathcal{L}}
\def\cM{\mathcal{M}}
\def\cO{\mathcal{O}}
\def\cP{\mathcal{P}}
\def\cQ{\mathcal{Q}}
\def\cV{\mathcal{V}}
\def\cS{\mathcal{S}}
\def\cT{\mathcal{T}}
\def\cX{\mathcal{X}}

\title[Geometry of information: classical and quantum aspect]{GEOMETRY OF  INFORMATION: \vspace{5pt}\\
CLASSICAL AND QUANTUM ASPECTS}
\author{No\'emie Combe}
\address{No\'emie Combe, Max Planck Institut for Mathematics in Sciences,  Inselstra\ss e 22, 04103 Leipzig, Germany}
\email{noemie.combe@mis.mpg.de}

\author{Yuri I. Manin}
\address{Yuri I. Manin, Max Planck Institut for Mathematics, Vivatsgasse  7, 53111 Bonn, Germany}
\email{manin@mpim-bonn.mpg.de}

\author{Matilde Marcolli}
\address{Matilde Marcolli. Department, Mail Code 253-37, Caltech, 1200 E. California Blvd., Pasadena, CA 91125, USA}
\email{matilde@caltech.edu}

\date{}

\makeatletter
\@namedef{subjclassname@2020}{\textup{2020} Mathematics Subject Classification}
\makeatother

\keywords{
}
\subjclass[2020]{Primary: 14G10; 18F30 Secondary: 55R99}

\thanks{N. C. Combe acknowledges support from the Minerva Fast track grant from the Max Planck Institute for Mathematics in the Sciences, in Leipzig.\\ M. Marcolli acknowledges support from NSF grants DMS-1707882 and DMS-2104330}

\begin{document}
\maketitle
\begin{abstract}
 In this article, we describe various aspects of categorification
of the structures, appearing in information theory. These aspects include
probabilistic models both of classical and quantum physics,
emergence of $F$--manifolds, and motivic enrichments.
\end{abstract}

\tableofcontents

\section{Introduction and summary}


This paper interconnects three earlier works by its (co)authors: \cite{Ma99},  \cite{Mar19},
and \cite{CoMa20}.

\vspace{3pt}

The monograph \cite{Ma99} surveyed various versions and applications
of the notion of varieties $\cX$, whose tangent sheaf $\cT_{\cX}$ is endowed with
{\it a commutative, associative and $\cO_{\cX}$--bilinear multiplication}.  Such $\cX$ 
got a generic name {\it F--manifolds}
by that time. Attention of geometers and specialists in mathematical physics
was drawn to them, in particular, because many deformation spaces
of various origin are endowed with natural
$F$--structures.

\vspace{3pt}

In \cite{CoMa20}, it was observed, that geometry
of spaces of probability distributions on finite sets, ``geometry of information'', that was
developing independently for several decades, also
led to some classes of $F$--manifolds.

\vspace{3pt}

The research in \cite{Mar19} extended these constructions to the domain of
{\it quantum probability distributions}, that we call ``geometry of quantum information''
here.

\vspace{3pt}

In this paper, we are studying  categorical encodings of geometries of classical
and quantum information,
including its $F$--manifolds facets, appearing in Sec.~\ref{S:4}, and further developed in Sec.~\ref{S:5}.
 Sec.~\ref{S:1} is a brief survey of both geometries.
In the categorical encoding we stress the aspects related to the
{\it monoidal structures and dualities} of the relevant categories which we survey in Sec.~\ref{S:2} and \ref{S:3}.
Finally, in Sec.~\ref{S:6} we introduce and study high level categorifications
of information geometries lifting it to the level of motives.

\section{Classical and quantum probability distributions}\label{S:1}

\subsection{Classical probability distributions on finite sets.}(\cite{CoMa20},\,3.2) \label{S:1.1}

Let $X$ be a finite set. Denote by $R^X$ the $\bR$--linear space of functions $(p_x): X\to \bR$.

\vspace{3pt}

By definition, {\it a classical probability distribution on $X$} is a point
of the simplex $\Delta_X$  in $R^X$ spanned by  the end--points of 
basic coordinate vectors in  $R^X$:
\[
\Delta_X:= \{ (p_x) \in R^X \,|\, p_x\ge 0,  \sum_x p_x =1 \}.
\]

We denote by ${}^{\circ} \Delta_X$ its maximal open subset
\[
{}^{\circ} \Delta_X:= \{ (p_x) \in \Delta_X \,|\, all\ p_x> 0 \}.
\]

Sometimes it is useful to replace $\Delta _X$ by ${}^{\circ} \Delta_X$  in the definition above.
Spaces of distributions become subspaces of {\it cones}.  
For a general discussion of the geometry of {\it cones}, see Sec. 3 below. 

\vspace{3pt}
\subsection{Quantum probability distributions on finite sets.}\label{S:1.2}

The least restrictive environment, in which we can define quantum probability distributions
on a finite set, according to Sec. 8 of [Mar19], involves an additional choice of
finite dimensional Hilbert space $\cV$. This means that $\cV$ is
finite dimensional vector space over the field of complex numbers $\bC$,
endowed with a scalar product $\langle v,v^{\prime}\rangle \, \in \bC$ such that
for any $v, v^{\prime} \in \cV$ and $a\in \bC$ we have 
\[
 \langle av, v^{\prime} \rangle = a\langle v,v^{\prime}\rangle, \  \langle v, av^{\prime} \rangle = \overline{a} \langle v, v^{\prime} \rangle.
\]
Here $a\mapsto \overline{a}$ means complex conjugation.

\vspace{3pt}

In particular, $\langle , \rangle$ is $\bR$--bilinear.

\vspace{3pt}

Whenever $\cV$ is chosen, we can define for any finite set $X$
the finite dimensional Hilbert space $\cH_X := \oplus_{x\in  X} \cV$, the direct
sum of $\text{card}\, X$ copies of $\cV$. 

\vspace{3pt}
 Finally, {\it a quantum  probability distribution on $X$} is a linear operator
 $\rho_X : \cH_X \to \cH_X $ such that $\rho_X=\rho_X^*$ and $\rho_x \ge 0.$
Here $*$ is the Hermitian conjugation. Such an operator is also called
{\it a density matrix.}

\vspace{3pt}

{\it Remarks.} The space $\cV$ represents the quantum space of internal degrees of freedom
of one point $x$. Its choice may be motivated by physical considerations,
if we model some physical systems. Mathematically, different choices of $\cV$ 
may be preferable when we pass to the study of categorifications: cf. below.

\vspace{3pt}
\subsection{Categories of classical probability distributions}  (\cite{Mar19}, Sec. 2)\label{S:1.3}. 

Let $Y,X$ be two
finite sets. Consider the real linear space $R^{Y,X}$ consisting of maps $Y\times X \to \bR$:
$(y,x) \mapsto S_{yx}$.

\vspace{3pt}
Such a map is called {\it a stochastic matrix}, if

\vspace{3pt}

(i) $S_{yx} \ge 0$ for all $(y,x)$.

\vspace{3pt}

(ii) $\sum_{y\in Y} S_{yx} =1$ for all $x\in X$.

\vspace{3pt}

\begin{proposition} \label{P:1.3.1}

  Consider pairs $(X,P)$, consisting of a finite set $X$ and 
a probability distribution $(p_x)$ on $X$ (one point of the closed set of probability distributions,
as above). 

\vspace{3pt}

These pairs are objects of the category $\cF\cP$, morphisms in which
are stochastic matrices 
\[
Hom_{\cF\cP}((X,P),(Y,Q)) := (S_{yx}),
\]
They are related to the distributions $P$, $Q$ by the formula $Q=SP$, i. e. 
\[
q_y = \sum_{x\in X} S_{yx} p_x ,
\]
where $(q_y)$ is the classical probability distribution,
assigning to $y\in Y$ the probability $q_y$.

Composition of morphisms is given by matrix multiplication.
\end{proposition}

\vspace{3pt}

Checking correctness of this definition is a rather straightforward task.
In particular, for any $(p_x) \in P$ and any stochastic matrix $(S_{yx})$,
$q_y \ge 0$, and 
\[
\sum_{y\in Y} q_y = \sum_{x\in X} p_x \sum_{y\in Y} S_{yx} = \sum_{x\in X}p_x =1,
\]
so that $SP$ is a probability distribution.

\vspace{3pt}

Associativity of composition of morphisms follows from associativity
of matrix multiplication.

\vspace{3pt}
\subsection {Categories of quantum probability distributions} (\cite{Mar19}, Sec. 8)\label{S:1.4}.

We now pass to the quantum analogs of these notions. 

\vspace{3pt}

Let again  $X$ be a finite set, now endowed as above,
with a finite dimensional 
Hilbert space $\cH_X = \oplus_{x\in X} \cV$ and a 
density matrix $\rho_x : \cH_X \to \cH_X$.

\vspace{3pt}

Given a finite dimensional Hilbert space $\cV\simeq \bC^N$, consider
an algebra $B(\cV)$  of linear operators on this space,
 containing the convex set   $M_{\cV}$      of 
density operators/matrices $\rho$ as in Sec.\ref{S:1.2} above,
{\it satisfying the additional condition}  $Tr(\rho)=1$.

\vspace{3pt}

A linear map $\Phi: B(\cV) \to B(\cV)$ is called {\it positive} if it
maps positive elements $\rho\geq 0$ in $B(\cV)$ to positive elements 
and it is {\it completely positive} if for all $k\geq 0$ the
operator $\Phi\otimes Id_k$ is positive on $B(\cV) \otimes M_k(\bC)$. Completely
positive maps form a cone $\cC\cP_{\cV}$: for all relevant information
regarding cones, see Sec.~\ref{S:3} below.

\vspace{3pt}

 A {\it quantum channel} is a
trace preserving completely positive map $\Phi: M_{\cV} \to M_{\cV}$. 
Composition of quantum channels is clearly again a quantum channel.
A quantum channel $\Phi$ can be represented by a matrix, the Choi matrix $S_\Phi$
which is obtained by writing $\rho'=\Phi(\rho)$ in the form
\[ 
\rho'_{ij}= \sum_{ab} (S_\Phi)_{ij,ab} \rho_{ab}. 
\]
The first pair of indices of $(S_\Phi)_{ij,ab}$ defines the row of the matrix $S_\Phi$ and
the second pair the column. Because quantum channels behave well under
composition, they can be used to define morphisms of a category of finite quantum probabilities.

\vspace{3pt}

A quantum analog of the respective statistical matrix is the so called
{\it stochastic Choi matrix $S_Q$} (see \cite{Mar19}, (8.1), where
we replaced Marcolli's notation $S_{\Phi}$ by
our $S_Q$, with $Q$ for quantum).

\vspace{3pt}

These triples are objects of the category $\cF\cQ$, morphisms in which
can be represented by stochastic Choi matrices so that
\[
(\rho_Y)_{ij} = \sum_{ab} (S_Q)_{ij,ab}(\rho_X)_{ab} .
\]

\vspace{3pt}

We have  omitted here a description of the  encoding of stochastic Choi
matrices involving the choices of bases in appropriate vector
spaces, and checking the compatibility of bases changes
with composition of morphisms. 

\vspace{3pt}

As soon as one accepts this,
the formal justification of this definition can be done in the same way
as that of Proposition \ref{P:1.3.1}.

\vspace{3pt}

\subsection{Monoidal categories} \label{S:1.5}

Speaking about monoidal categories, we adopt  basic definitions, axiomatics,
and first results about categories, sites, sheaves, and their homological and homotopical properties
developed in \cite{KS06}. In particular, sets of objects and morphisms of a category always will be
{\it small sets} (\cite{KS06}, p. 10).

\vspace{3pt}

Sometimes we have to slightly change terminology, starting with monoidal categories themselves.
We will call a {\it monoidal category} here a family of data, called a {\it tensor category}
 during entire Chapter 4 of \cite{KS06} and the rest of the book, with exception of two lines in remark 4.2.17, p. 102.
 
\vspace{3pt}
 
 According to the Definition 4.2.1. of \cite{KS06} (p. 96),
 a monoidal category is a triple ($\cH, \otimes , a$),
 where $\cH$ is a category, $\otimes$ a bifunctor 
 $\cH\times \cH \to \cH$, and $a$ is an ``associativity'' isomorphism
 of triple functors, constructed form $\otimes$ by two different bracketings.
 
\vspace{3pt}
 
 The asociativity isomorphism must fit into the commutative diagram (4.2.1) 
 on p. 96 of \cite{KS06}.
 
\vspace{3pt}
 
  According to the Def. 4.2.5, p. 98 of \cite{KS06},  {\it a unit object $\b1$}
 of a monoidal category  $\cH$ is an object, endowed with an isomorphism
 $\rho : \b1 \otimes \b1\to \b1$ such that the functors 
 $X\mapsto X\otimes \b1$ and $X\mapsto \b1$ are  fully faithful.
 By default,  our monoidal categories, or their appropriate versions, will be endowed by  unit objects.

\vspace{3pt}
 
 Lemma 4.2.6 of \cite{KS06}, pp. 98--100, collects all natural compatibility relations
 between the monoidal multiplication $\otimes$ and the unit object $(\b1, \rho )$, categorifying
 the standard properties of units in set--theoretical monoids.

\vspace{3pt}

\subsection{Duality in monoidal environments.} \label{S:1.6}

The remaining part of this Section contains a brief review, based upon \cite{Ma17},  of categorical aspects
of monoidality related to dualities between monoidal categories with units.

\smallskip

They must be essential also for the understanding of quantum probability distributions,
because generally the relevant constructions appeared during the
study of various quantum models: see references in \cite{Ma88}, \cite{Ma17}, and
a later development \cite{MaVa20}.

\vspace{3pt}

Let $(\cH, \bullet , \b1)$ be a monoidal category, and
let $K$ be an object of $\cH$. (Notice that here we changed the notation
of the monoidal product, earlier $\otimes$,
and replaced it by $\bullet$).

\vspace{3pt}

\begin{definition} \label{D:1.6.1} 

 A functor  $D_K:\ \cH \to \cH^{op}$
is a duality functor, if it is an antiequivalence of categories, such
that for each object $Y$ of  $\cH$ the functor 
\[
X\mapsto  Hom_{\cH} (X\bullet Y, K)
\]
is representable by $D_K(Y)$.

In this case $K$ is called a dualizing object.

\end{definition}
\smallskip

(i) The data   $(\cH , \bullet , \mathbf{1} , K) $  are called
{\it a Grothendieck--Verdier} (or GV--){\it category}.

\medskip

 (ii) $(\cH, \circ, K)$ {is a monoidal category with unit object $K$.}

\medskip

\subsection{Example: quadratic algebras}\label{S:1.7}

 Let $k$ be a field.
A {\it quadratic algebra} is defined as an associative graded algebra
$A = \oplus_{i=0}^{\infty} A_i$ generated by $A_1$ over $A_0=k$,
and such that the ideal of all relations between generators $A_1$
is generated by the subspace of quadratic relations $R(A)\subset A_1^{\otimes 2} $.

\smallskip

Quadratic algebras form objects of a category $\cQ\cA$, morphisms in which
are homomorphisms of graded algebras identical on terms of degree $0$.
It follows that morphisms $f: A\to B$ are in canonical bujection
between such linear maps $f_1: A_1\to B_1$ for which
$(f_1\otimes f_1) (R(A)) \subset R(B)$. 

\smallskip

The main motivation for this definition was a discovery, that
if in the study of a large class of quantum groups we replace (formal)
deformations of the universal enveloping algebras of the relevant Lie algebras by (algebraic) deformations
of the respective algebras of functions, then in many cases we land in the category $\cQ\cA$.

\vspace{3pt}

The monoidal product $\bullet$ in $\cQ\cA$ can be introduced directly
as a lift of the tensor product of linear spaces of generators: $A\bullet B$ is generated
by $A_1\otimes B_1$, and its space of quadratic relations is
$S_{23}(R(A)\otimes R(B))$, where the permutation map
$S_{23}: A_1^{\otimes 2} \otimes B_1^{\otimes 2} \to (A_1\otimes B_1)^{\otimes 2} $
sends $a_1 \otimes a_2 \otimes b_1 \otimes b_2$ 
to $a_1 \otimes b_1 \otimes a_2 \otimes b_2$.

\vspace{5pt}

{\it Remark.} Slightly generalizing these definition, we may assume that
our algebras are $\bZ \times \bZ_2$--graded, that is {\it supergraded}.
This will lead to the appearance of additional signs $\pm$ in various places.
In particular, in the definition of $S_{23}$ there will be sign $-$, if both
$a_2$ and $b_1$ are odd.

\smallskip

This might become very essential in the study of quantum probability distributions where 
physical motivation comes from models of fermionic lattices.

\smallskip

We return now to our category $\cQ\cA$.

\smallskip

Define the dualization in $\cQ\cA$ as
a functor $A\mapsto A^!$ extending the linear dualization 
$A_1\mapsto A_1^* := Hom (A_1, k)$ on
the spaces of generators. The respective subspace of quadratic
relations will be $R(A)^{\bot} \subset (A_1^*)^{\otimes 2}$, the
orthogonal complement to $R(A)$.

\vspace{3pt}

\begin{proposition} \label{P:1.7.1}
 Consider $\b1 := k[\tau]/(\tau^2)$ as the object
of $\cQ\cA$, and put $K:= k[t]$.

\vspace{3pt}

Then $(\cQ\cA, \bullet , \b1, K)$ is a GV--category. 

\vspace{3pt}

More precisely, in the respective dual category $(\cQ\cA^{op} ,\circ , K, \b1  =K^!)$
the ``white product'' $\circ$ is another lift of the tensor product
of linear generators $A_1\otimes B_1$, with quadratic relations
$S_{23}(R(A)\otimes B_1^{\otimes 2} + A_1^{\otimes 2} \otimes R(B))$. 
\end{proposition}

\vspace{3pt}

For a proof, see \cite{Ma88}, pp. 19--28.

\section {Monoidal duality in categories of classical probability distributions}\label{S:2}

\medskip
\subsection{Generalities} \label{S:2.1}

According to \cite{KS06}, (Examples 4.2.2 (vi) and (v), p.96),
if a category $\cC$ admits finite products $\times$, resp. finite coproducts $\sqcup$,
then both of them define monoidal structures on this category.

\smallskip

Both products and coproducts are defined by their universal
properties  in  the Def. 2.2.1, p. 43, of \cite{KS06}.

\smallskip

If \, $\cC$ admits finite inductive limits and finite projective limits,
then it has an initial object $\emptyset_{\cC} $, which is
the unit object for the monoidal product $\sqcup$. Any morphism
$X\to \emptyset_{\cC}$ is an isomorphism, and 
$\emptyset_{\cC} \times X \simeq \emptyset_{\cC}$
(\cite{KS06}, Exercise 2.26, p. 69).

\smallskip

When studying monoidal dualities  in various
categories of structured sets, it is useful to keep in mind
the following archetypal example.  For any set $U$,
denote by $\cP(U)$ the set of  all non--empty subsets
of $U$. Then, for any two sets without common elements $X$, $Y$,
 there exists a natural bijection

\begin{equation}\label{E:2.1}
\cP (X\cup Y) \to (\cP(X) \times  \cP(Y)) \cup \cP(X)  \cup \cP (Y).
\end{equation}
Namely, if a non--empty subset $Z\subset X\cup Y$ has  empty intersection with $Y$,
resp. $X$, it produces the last two terms in the r.h.s. of \eqref{E:2.1}.
Otherwise, it produces a pair of non--empty subsets $Z\cap X \in \cP(X)$
and $Z\cap Y \in \cP(Y)$.

\smallskip

A more convenient version of \eqref{E:2.1} can be obtained, if one works in a category $\cS$
of pointed sets, and defines the set of subsets $\cP (X)$ as previously, but adding
to it the empty subset as the marked point. Then the union $\cup$ in \eqref{E:2.1}
is replaced by the smash product $\vee$, and the map \eqref{E:2.1} extended to
a bifunctor of $X,Y$, becomes a ``categorification'' of the formula for an exponential map 
$e^{x+y} = e^x\cdot e^y$ underlying transitions between classical models
in physics and quantum ones. In particular, \eqref{E:2.1} connects the unit element
for direct product with the unit/zero element for the coproduct/smash product.

\smallskip

We will now describe a version of these constructions applied to sets
of probability distributions.

\medskip

\subsubsection{\bf A warning} \label{S:2.1.1}

 If we construct a functor $\cC \to \cD $
or $\cC \to \cD^{op}$, sending $\times$ to $\sqcup$
and exchanging their units $\emptyset_{\cC}$ and $\b1_{\cD}$,
it {\it cannot be completed to a duality functor}  in the sense of 
Def.\ref{D:1.6.1}  above: a simple count of cardinalities shows it.

\smallskip

However, we will still consider such functors as weaker versions of
monoidal dualities, and will not warn a reader about it anymore.

\medskip
\subsection{Category of classical probability distributions.}  

If the cardinality
of a finite set $X$ is one, there is only one classical probability distribution
on $X = \{x\}$, namely $p_x =1$. In \cite{Mar19}, Sec.2, such objects are called
singletons. 

\smallskip

Singletons are {\it zero objects} in $\cF\cP$ (\cite{Mar19}, Lemma 2.5),
that is, they have the same categorical properties as the objects $\emptyset_{\cC}$,
described in Sec.\ref{S:2.1}.

\smallskip

Morphisms  that factor through zero objects are generally called zero morphisms.
In $\cF\cP$, they are explicitly described as ``target morphisms'' by
\cite{Mar19}, 2.1.2: they are such morphisms $\hat{Q} : (X,P) \to (Y,Q)$, for which
$\hat{Q}_{ba} =  Q_b$ .

\smallskip

Category $\cF\cP$ is not large enough for us to be able to use
essential constructions and results from [KS06], related to
multiplicative/additive transitions sketched in Sec.\ref{S:2.1} above.

\smallskip
M. Marcolli somewhat enlarged it by replacing
in the definition of objects of $\cF\cP$ 
finite sets by  pointed finite sets. 
 In \cite{Mar19}, the resulting category
 is denoted $\cP\cS_*$,
so that  $\cF\cP$ is embedded in it.
Objects, morphisms sets in $\cP\cS_*$, their compositions etc.
are explicitly described in Def. 2.8, 2.9, 2.10 of \cite{Mar19}.
Objects of $\cP\cS_*$ are called {\it probabilistic pointed sets}
in Def.2.8.

\smallskip

Besides embedding,
there is also a {\it forgetful functor} $\cP\cS_* \to \cP\cF$
(\cite{Mar19}, Remark 2.11).

\smallskip
This becomes a particular case of
constructions in the categories of small sets in \cite{KS06}, Ch.~\ref{S:1}. 

\medskip

\subsubsection{\bf Coproducts of probabilistic pointed sets and classical probability
distributions.} \label{S:2.2.1}
Start with the usual smash product of pointed sets
\[
(X,x) \vee (Y,y) := ((X\sqcup Y)/(X\times {y} \cup x\times Y), *) ,
\]
where $*$ is the ``smash'' of the union of two coordinate axes
$X\times y\cup x\times Y$. It induces on probabilistic pointed sets
obtained from finite probability distributions the product
of statistically independent probabilities.
\[
(X,P) \sqcup (Y, Q) = (X\times Y, \, p_{(x,y)} =  p_xq_y )  .
\]
(\cite{Mar19}, Lemma 2.14) .

\smallskip

We will denote by $\emptyset_{\cF}$ any object of $\cP\cS_{*}$, 
consisting of a finite set and a point in it with prescribed probability 1.
All these objects are isomorphic.

\medskip


\begin{theorem} \label{Th:2.3.1}
\ 

 (i) The triple $(\cP\cS_*, \sqcup , \emptyset_{\cF} )$  is a monoidal category with unit.

\smallskip

(ii) The triple $(\cF\cP , \times , \{pt\})$ is a monoidal
category with unit. 
\end{theorem}
\smallskip

For detailed proofs, see \cite{Mar19}, Sec. 2.

\medskip

\subsection{A generalization}
 Let $(\cC, \sqcup , {0}_{\cC})$
be a monoidal category with unit/zero object.

\smallskip

Generalizing the passage from $\cS_*$ to $\cP\cS_*$,
M. Marcolli defines the category $\cP\cC$, a probabilistic
version of $\cC$ (\cite{Mar19}, Def. 2.18).

\smallskip

One object of $\cP\cC$ is a formal finite linear combination
$\Lambda C := \sum_i \, \lambda_i C_i$, where $C_i$ are objects of $\cC$.
\smallskip

\smallskip

One morphism $\Phi : \Lambda C \to \Lambda^{\prime} C^{\prime}$
is a pair $(S,F)$, where $S$ is a stochastic matrix with $S\Lambda = \Lambda^{\prime}$,
and $F_{ab,r}: C_b \to C_a^{\prime}$ are real numbers in $[0,1]$
such that $\sum_r\mu_r^{ab} =S_{ab}$.

\smallskip

As before, one can explicitly define a categorical coproduct $\sqcup$ in $\cP\cC$,
so that it becomes a monoidal category with zero object.

\medskip
\section{Monoidal duality in categories of quantum probability distributions} \label{S:3}

\medskip
\subsection{Category of quantum probability distributions.}  \label{S:3.1}

We now
return to the category $\cF\cQ$ of quantum probability
distributions.

\smallskip

We will be discussing the relevant versions of monoidal dualities
for enrichments of $\cF\cQ$, based upon variable
categories $\cC$.

\medskip

\begin{definition} (\cite{Mar19}, Def. 8.2). \label{3.1.1}

  Given a category $\cC$, its
quantum probabilistic version $\cQ\cC$ is defined as follows.

\smallskip

One object of $\cQ\cC$ is a finite family $((C_a,C_b), \rho:= (\rho_{ab}))$,
where
\[
(C_a, C_b) \in \textrm{ Ob}\,C \times \textrm{Ob} \, C , \quad a, b =1,\dots , N\ge 1,
\]
and $\rho = (\rho_{ab})$ is the Choi matrix of a quantum channel as in Sec. \ref{S:1.4} above.

\smallskip

One morphism between two such objects,  $(( C_a,C_b), \rho ) $ (source)
and $(( C_a^{'},C_b^{'} ),\rho^{'})$ (target)  is given by a Choi matrix, as in Sec.~\ref{S:1.4}  above,
entries of which now are morphisms in $\cC$.

\smallskip

Composition of two morphisms is defined similarly to the classical case,
so that the usual associativity diagrams lift to $\cQ\cC$.
\end{definition}
\medskip

\subsection{Monoidal structures.}\label{S:3.2}

 Assume now that $\cC$ is endowed with 
a monoidal structure with unit/zero object. Then it can be lifted
to $\cQ\cC$ in the same way as in classical case: see
\cite{Mar19}, Proposition 8.5.

\smallskip

We can therefore extend the (weak) duality formalism of Sec. 2 to the
case of quantum probability distributions, keeping in mind
the warning stated in Sec.~\ref{S:2.1.1}. 

\vspace{5pt}
\section{Convex cones and $F$--manifolds}\label{S:4}

\vspace{3pt}

\subsection{Convex cones} \label{S:4.1}  

Let $R$ be a finite--dimensional 
real linear space. A  non--empty subset $V\subset R$ is called {\it a cone},
if 

\vspace{5pt}

(i) $V$ is closed with respect to addition and multiplication by positive reals;

\vspace{5pt}

(ii) The topological closure of $V$ does not contain a real affine subspace $R$ of positive dimension.

\vspace{5pt}

{\it Basic example.} Let $C\subset R$ be a convex open subset of $R$ whose closure 
does not contain $0$.
Then the union of all half--lines in $R$, connecting $0$ with a point of $C$,
is an open cone in $R$.

\vspace{3pt}

Here convexity of $C$ means that any segment of real line connecting two different
points of $C$, is contained in $C$.

\vspace{5pt}

 \subsection{Convex cones of probability distributions on finite sets.} (\cite{CoMa20}, 3.2.) \label{S:4.2}
 
     Let $X$ be a finite set. As in \ref{S:1.1} above, we start with the real linear space $R^X$ and denote
by $V$ the union of all oriented real half--lines from zero to one of the points
of $\Delta_X$, or else of ${}^{\circ} \Delta_X$.

\vspace{3pt}

Such cones $V$ are called {\it open}, resp. {\it closed}, cones of classical
probability distributions on $X$.

\vspace{5pt}

\subsection{Characteristic functions of convex cones.}\label{S:4.3}

 Given a convex cone $V$ in finite--dimensional real linear space $R$,
construct its {\it characteristic function} $\varphi_V : V\to \bR$
in the following way.

\vspace{3pt}

Let $R^{\prime}$ be the dual linear space, and $\langle x,x^{\prime}\rangle$
the canonical scalar product between $x\in R$ and $x^{\prime}\in R^{\prime}$.
Choose also a volume form $vol^{\prime}$ on $R^{\prime}$ invariant wrt
translations in $R^{\prime}$. Then put
\[
\varphi_V(x) := \int e^{-{\langle x,x^{\prime}\rangle}} vol^{\prime} .
\]

\subsubsection{\bf Claim.} 

{\it Such a volume form and a characteristic function are defined up to
constant positive factor.}

\vspace{3pt}

This follows almost directly from the definition.

\vspace{3pt}

 Now we focus on the  differential geometry of convexity.
 
 \smallskip
 
A cone $V$ is a smooth manifold, its tangent bundle $\cT_V$
can be canonically trivialised, $\cT_V =V\times R$,
in the following way:  $\cT_{V,x}$ is identified with $R$
by the parallel transport sending $x\in V$ to $0\in R$. Choose
an affine coordinate system $(x^i)$ in $R$ and put
\[
g_{ij}:=\frac{\partial^2\, \ln\, \varphi_V}{\partial_{x^i}\partial_{x^j}} .
\]

\medskip

\begin{theorem} ([Vi63]). \label{Th:4.4}

The metric $\sum_{i,j}g_{ij}dx^idx^j$ is a Riemannian metric on $V$.
The associated torsionless canonical connection on $\cT_V$ has components
\[
\Gamma^i_{jk} = \frac{1}{2} \sum_l g^{il}\frac{\partial^3\,\ln\,\varphi}{\partial{x^j}\partial{x^k}\partial{x^l}}\, ,
\]
where $(g^{ij})$ are defined by $\sum_jg^{ij}g_{jk} = \delta^i_k .$

\vspace{3pt}
Therefore, by putting 
\[
\sum_ja^j\partial_{x_j} \circ \sum_k b^k\partial_{x^k} := \sum_{ijk} \Gamma^i_{jk} a^jb^k\partial_{x^i},
\]
we define on $\cT_V$ a commutative $\bR$--bilinear composition.
\end{theorem}

\vspace{5pt}

 \subsection{Convex cones of stochastic matrices and categories of classical probability distributions}  (\cite{Mar19}, Sec. 2).\label{S:4.5}
 
  Our first main example are cones of stochastic matrices from Sec.\ref{S:1.3} 
  above.

\vspace{3pt}

Let $Y,X$ be two
finite sets. Consider the real linear space $R^{Y,X}$ consisting of maps $Y\times X \to \bR$:
$(y,x) \mapsto S_{yx}$, where $S$ is a stochastic  matrix. As explained above,
such stochastic matrices can be considered as morphisms in the category
$\cF\cP$.

\vspace{3pt}

\begin{proposition} \label{P:4.5.1}
 The sets of morphisms $Hom_{\cF\cP}((X,P),(Y,Q))$ are convex sets.
\end{proposition}

\vspace{3pt}

 \subsection{Convex cones of quantum probability distributions on finite sets.} (\cite{Mar19}, Sec. 8 and others.)\label{S:4.6}
Using now stochastic Choi matrices, encoding morphisms between
objects of the category 
 $\cF\cQ$, as explained in Sec.\ref{S:1.4} above,
we will prove the following quantum analog of the Proposition \ref{P:4.5.1}.

\vspace{3pt}

\begin{proposition}\label{P:4.6.1}

 The sets of morphisms in $\cF\cQ$ are convex sets.
\end{proposition}

\vspace{3pt}

We leave both proofs as exercises for the reader.

\vspace{5pt}

\subsubsection{Remarks: Comparison between classical and quantum probability
distributions.} \label{S:4.6.2}
\ 

(i) The $\bR$--linear spaces $R$ or $R^X$ from above
correspond to $\bC$--linear spaces $\cH$ or $\cH_X$
from \cite{Mar19}, Def. 8.1.

\vspace{3pt}

(ii) The $\bR$--dual spaces $R^{\prime}$ correspond to
the $\bC$--{\it antidual} spaces $\cH^*$ from \cite{Mar19}.

\vspace{3pt}

The real duality pairing $\langle x,x^{\prime} \rangle$ is replaced
by the complex antiduality pairing: this means that,
for $h\in \cH, h^*\in \cH$ and $a\in \bC$,
we have  
\[
\langle ax,x^{\prime} \rangle =
a\langle x,x^{\prime} \rangle, \quad\quad
\langle x,ax^{\prime} \rangle =
\overline{a}\langle x,x^{\prime} \rangle,
\]
where $a\mapsto \overline{a}$ is the complex conjugation map.

\vspace{3pt}

(iii) This implies that the direct sum of real spaces $R\oplus R^{\prime}$
on the classical side which is replaced by $\cH \oplus \cH^*$
on the quantum side, can be compared with the direct sum
of also {\it real} subspaces of $\cH \oplus \cH^*$
corresponding to the {\it eigenvalues} $\pm i$ of the operator
combining $h\mapsto h^*$.

\vspace{3pt}

A parametric variation of this structure should lead to
{\it paracomplex geometry}, which entered the framework
of geometry of quantum information in Sec.4 of \cite{CoMa20}. 
The algebra of paracomplex numbers (cf.~\cite{Ya}) is defined as the real vector space $\fC = \R\oplus \R$ with the multiplication
\[
(x,y) \cdot (x',y') = (xx' + yy', xy' + yx'). 
\]
Put $\varepsilon : =(0,1)$.   Then $\varepsilon^2 =1$, and moreover
\[
\fC =\R+\e\R= \{z=x+\e y \, |\, x,y \in \R \}.
\]

Given a paracomplex number $z = x+\varepsilon y$, its conjugate is defined by $\overline{z}:= x-\e y$.
The paracomplex numbers form a commutative ring with characteristic 0. Naturally, arises the notion paracomplex structure on a vector space. 

\vskip.2cm

The paracomplex structure enters naturally the scene of the manifold of probability distributions over a finite set, and more generally to the case of statistical manifolds related to exponential families. Indeed, these real manifolds are identified to a projective space over the algebra of paracomplex numbers (see Proposition 5.9 in\, \cite{CoMa20}). It is well known that this manifold of probability distributions is endowed with a pair of affine, dual connections $\nabla$ and $\nabla^*$. So, the underlying affine symmetric space is defined over a Jordan algebra which is generated by $\{1,\e\}$ and verifies $\e^2=1$ or $\e^2=-1$. This manifold not being complex, the paracomplex case remains the only possibility (Proposition 5.4 \cite{CoMa20}). 

\vspace{3pt}

There are interesting questions related to the occurrence of paracomplex $F$-manifolds in
information geometry. For instance, the question of a classification of these paracomplex structures,
along the lines of recent classification results for small dimensional $F$-manifolds over $\C$ by
Hertling and collaborators, or of Dubrovin's analytic theory.

\vspace{5pt}
\section{Clifford algebras and Frobenius manifolds}\label{S:5}

\vspace{3pt}

\subsection{Hilbert spaces over Frobenius algebras.}\label{S:5.1}

 Let $\cA$ be a  commutative algebra over $\bR$
of finite dimension $n$, 
generated by $n$ linearly independent elements 
$B_{1},\dots, B_q$ satisfying
relations  $B_i\cdot B_j= \gamma_{ij}^{k}B_{k}$.  

\vspace{3pt}

Moreover assume, that $\cA$ is endowed with a nondegenerate 
bilinear form $\sigma$ (Frobenius form), satisfying the associativity
property $\sigma(ab,c) =\sigma (a, bc)$, and a homomorphism $\eta : \cA\to \bR$,
whose kernel contains no non--zero left ideal of $\cA$.

\vspace{3pt}

One can see that then $\sigma_{ij}:= \sigma(B_iB_j) = \sum_{s}\gamma_{ij}^s\eta_s$, where 
$\eta_s\in \bR$.

\vspace{3pt}

As usual, in such cases we will omit summation over repeated indices and
write the r.h.s. simply as $\gamma_{ij}^s\eta_s$.
\smallskip

Denote by $(B^b)$ the dual basis to $(B_a)$ with respect to $\sigma$.

\vspace{3pt}

Now consider a right free  $\cA$-module $\cM\, (\cA)$ of rank $r$. It has a natural structure of real $rn$--dimensional
linear space. It can be also represented as the space of matrices
over $\cA$.

\smallskip

Generally, below we will be considering Hilbert spaces $\cM$, endowed with a compatible action 
of $\cA$. 

\vspace{5pt}

\subsubsection{Example}\label{E:5.1.1}

 Consider a particular case  $q=2$.
One can see that there exist three different algebras
(up to isomorphism): {\it complex numbers, dual numbers}, and
{\it paracomplex numbers}.

\vspace{3pt}

The respective bases, denoted $(B_a)$ in Sec.~\ref{S:5.1} above, have
traditional notations: $(1, i)$, $i^2=-1$; $(1, \varepsilon )$,
$\varepsilon^2=0$; and $(1, \epsilon )$, $\epsilon^2=1$.

\vspace{3pt}

\subsection{General case: Clifford algebras} 

Let $k$ be a field of characteristic $\ne 2$;
$V$ a finite dimensional linear space over $k$; $Q: \, S^2(V) \to k$
a non--degenrate quadratic form on $V$. It defines the symmetric scalar product on $V$:
$\langle u,v\rangle := \frac{1}{2} [Q(u+v)- Q(u) -Q(V)]$.

\vspace{3pt}

\begin{definition} 

A Clifford algebra $Cl$ over $k$ is an associative unital $k$--algebra
of finite dimension $q$, endowed with generators $B_1,\dots ,B_q$
satisfying relations 
\[
B_iB_j+B_jB_i = 2\langle B_i, B_j \rangle .
\]
\end{definition}
\smallskip
Notice that $Q$ is an implicit part of the structure in this definition. 

\vspace{3pt}

If $Q=0$, then $C\ell$ is the exterior algebra of $V$ over $k$, hence
linear dimension of $C\ell$ is $2^n$, where $n= \mathrm{dim}\, V$ over $k$.
This formula holds also for $Q$ of arbitrary rank.

\vspace{3pt}

Finally, if $k=\bR$ and $Q$ is non--degenerate, it has a signature 
$(p,q)$, that is, $Q$ in an appropriate basis has the standard
form $v_1^2+ \dots +v_p^2 - v_{p+1}^2 - \dots -v_{p+q}^2$, hence $n=p+q$.
We will denote the respective Clifford algebras by $C\ell_{p,q}$.

\vspace{3pt}

This gives a complete list of (isomorphism classes of) Frobenius
$\bR$--algebras of finite dimension.

\vspace{3pt}

Clifford algebras have properties implying the existence of a symmetric scalar product on the vector space $V$. More precisely,
 $\langle u,v\rangle := \frac{1}{2} [Q(u+v)- Q(u) -Q(v)]$ where $Q$ is the quadratic form associated to the Clifford algebra. 
Using this definition, we can obtain the characteristic function, which is defined in  section 4.3.
 Recall that the characteristic function is explicitly  given by 
\[
\varphi(x) := \int e^{-{\langle x,x^{\prime}\rangle}} vol^{\prime},
\]
where  $R^{\prime}$ is the dual linear space, and $\langle x,x^{\prime}\rangle$ is
the canonical scalar product between $x\in R$ and $x^{\prime}\in R^{\prime}$, and $vol^{\prime}$ is a volume form on $R^{\prime}$ invariant w.r.t. 
translations in $R^{\prime}$. 
Therefore, one establishes a direct relation between those Clifford algebras and the characteristic functions defined in section~\ref{S:4.3}, and hence a relation to the $F$-manifolds. 

\vspace{3pt}

\subsection{The splitting theorem} \label{S:5.3}

Consider now  a stochastic matrix
(see subsections \ref{S:1.3} and  \ref{S:1.4} above)   acting upon a finite--dimensional
Hilbert space $\cH$.

\vspace{3pt}

\begin{theorem}\label{Th:5.3.1}
 $\cH$ has a canonical splitting into subsectors that are irreducible modules over respective
Clifford algebras.
\end{theorem}

\vspace{3pt}

\begin{proof}

  Let $M$ be a manifold, endowed with an affine flat structure, a compatible metric $g$, and an even symmetric rank 3 tensor $A$. Define a  multiplication operation $\circ$ on the tangent sheaf 
by $\circ:T_M\times T_M\to T_M$. The manifold $M$ is Frobenius if it satisfies local potentiality condition for $A$, i.e. locally everywhere there exists a potential function $\varphi$ such that $A(X,Y,Z)=\partial_{X,Y,Z}\varphi$, where $X,Y,Z$ are flat tangent fields and an associativity condition: $A(X,Y,Z)=g(X\circ Y, Z)=g(X, Y\circ Z)$ (see~\cite{Ma99}). 

\smallskip

Clifford algebras can be considered under the angle of matrix algebras as Frobenius algebras. They
 are equipped with a symmetric bilinear form $\sigma$, such that $\sigma(a\cdot b,c)=\sigma(a, b \cdot c)$. We consider a module over this Frobenius algebra $\cA$ and  construct the real linear space to which it is identified, denoted $E^{rq}$. 

\vspace{3pt}

Let us first discuss the rank 3 tensor $A$. We construct it on $E^{rq}$, using the $(p,q)$-tensor formula 
\[
T^{\alpha_1\dots \alpha_p}_{\beta_1\dots \beta_q}B_{a_1}\dots B_{a_p}B^{\beta_1}\dots B^{\beta_q}
\]
in the adapted basis.

\vspace{3pt}

There exists a compatible metric, inherited from the non-degenerate, symmetric  bilinear form defined on the algebra $\cA$, given by $\sigma:\cA\times\cA\to k$, where $\sigma(B_i,B_j)=\gamma_{ij}^kB_{k}$. Call this metric $g$.

\vspace{3pt}

Now that the rank 3 tensor $A$ and the metric $g$ have been introduced, we discuss the multiplication operation $``\circ"$. This multiplication operation is inherited from the multiplication on the algebra and given by $B_i\cdot B_j=\gamma_{ij}^k B_k$. It can be written explicitly by introducing a bilinear symmetric map $\overline{A}: E^{rq}\times E^{rq} \to E^{rq}$, which in local coordinates is  
\[
\overline{A}=A_{ab}^c=\sum_eA_{abe}g^{ec}, \quad g^{ab}=(g_{ab})^{-1}.
\]
Here $A_{abc} := g_{cm}A^m_{ab}$. 
The multiplication is thus defined by
\[
A g^{-1}: E^{rq} \times E^{rq}\to E^{rq},
\]
 with $X\circ Y=\overline{A}(X,Y)$ and $X,Y$ are local flat tangent fields. 

\vspace{3pt}

 Since we have defined the multiplication operation, we can verify the associativity property. Indeed, recall that the metric is inherited from the Frobenius form $\sigma$, which satisfies $\sigma(a\cdot b,c)=\sigma(a, b \cdot c)$. Naturally, this property is inherited on $E^{rq}$, where this associativity relation is given by $g(X\circ Y, Z)=g(X, Y\circ Z)$ for $X,Y,Z$ flat tangent fields. 
 
\vspace{3pt}

 Finally, by using the relation between $\overline{A}, A$ and $g$, we can see that the potentiality property is satisfied.
\end{proof}

\vspace{3pt}

See also Sec.~\ref{S:6.8} -- \ref{S:6.9} below.

\vspace{5pt}

\section{Motivic information geometry}\label{S:6}

\vspace{3pt}

This last section is dedicated to the construction of the highest  (so far)
floor of the Babel Tower of categorifications of probabilities.

\vspace{3pt}

We investigate possible extensions of some aspects of the formalism of information geometry to a motivic setting, represented by various types of Groth\-endieck rings.

\vspace{3pt}

A notion of motivic random variables was developed in \cite{Howe19}, \cite{Howe20}, based on
relative Grothendieck rings of varieties. In the setting
of motivic Poisson summation and motivic height zeta functions, as in \cite{Bilu18}, \cite{ChamLoe15},
\cite{CluLoe10}, \cite{HruKaz09}, one also considers other versions of the Grothendieck ring
of varieties, in particular the Grothendieck ring of varieties with exponentials. 
A notion of information measures for Grothendieck rings of varieties was introduced in \cite{Mar19b},
where an analog of the Shannon entropy, based on zeta functions, is shown to satisfy a suitable version of the
Khinchin axioms of information theory. We elaborate here  some of these ideas with the
goal of investigating motivic analogs of the Kullback--Leibler divergence and the Fisher--Rao information
metric used in the context of information geometry (see \cite{AmNag07}). 

\medskip

\subsection{Grothendieck ring with exponentials and relative entropy}\label{S:6.1}

We show here that, in the motivic setting, it is possible to implement a 
version of Kullback--Leibler divergence based on zeta functions,
using the Grothendieck ring of varieties with exponentials, defined in \cite{ChamLoe15}.

\vspace{3pt}

\begin{definition}\label{D:6.1.1}

The Grothendieck ring with exponentials
$KExp(\cV_K)$, over a field $K$, is generated by isomorphism classes of pairs
 $(X,f)$, where $X$ is  a $K$--variety,
and $f$ a morphism $f: X\to \bA^1$. Two such pairs $(X_1,f_1)$ and $(X_2,f_2)$ 
are isomorphic, if there is an isomorphism $u: X_1\to X_2$ of $K$--varieties such
that $f_1=f_2\circ u$. 

The relations in $KExp(\cV_K)$ are given by
\[
 [X,f]=[Y,f|_Y] + [U,f|_U], 
\]
for a closed subvariety $Y\hookrightarrow X$ and its open complement $U=X\setminus Y$, and
the additional relation
\[
 [X\times \bA^1, \pi_{\bA^1}]=0 
 \]
 where $\pi_{\bA^1}: X\times \bA^1 \to \bA^1$ is the projection on the second factor.

\vspace{3pt}
The ring structure is given by the product
\[
 [X_1,f_1]\cdot [X_2,f_2]= [X_1\times X_2, f_1 \circ \pi_{X_1} + f_2 \circ \pi_{X_2}] 
 \]
 where $f_1 \circ \pi_{X_1} + f_2 \circ \pi_{X_2}: (x_1,x_2)\mapsto f_1(x_1)+f_2(x_2)$.
\end{definition}

\vspace{3pt}

The original motivation for introducing the Grothendieck ring with exponentials was to
provide a motivic version of exponential sums. Indeed, for a variety $X$ over a finite field $\bF_q$,
with a morphism $f: X\to \bA^1$, a choice of character $\chi: \bF_q \to \bC^*$ determines
an exponential sum
\[ \sum_{x\in X(\cF_q)} \chi(f(x)) \]
of which the class $[X,f]\in KExp(\cV)_{\cF_q}$ is the motivic counterpart. The relation
$[X\times \bA^1, \pi_{\bA^1}]=[X,0]\cdot [\bA^1,id]=0$ corresponds to the property that, 
for any given character $\chi: \cF_q \to \cC^*$, one has $\sum_{a\in \bF_q}\chi(a)=0$.

\vspace{3pt}

Here we interpret the classes $[X,f]$ with $f: X \to \bA^1$ as pairs of a variety
and a potential (or Hamiltonian) $f: x\mapsto H_x=f(x)$. A family of commuting
Hamiltonians is represented in this setting by a class $[X\times \bA^1, F]$ with
$F: X\times \bA^1 \to \bA^1$, where for $\epsilon\in \bA^1\setminus \{ 0 \}$
the function $f_\epsilon: X \to \bA^1$ given by $f_\epsilon(x)=F(x,\epsilon)$ is
our Hamiltonian $f_\epsilon: x\mapsto H_x(\epsilon)=f_\epsilon(x)$. 

\vspace{3pt}

Note that, for this interpretation of classes $[X,f]$ as varieties with a potential (Hamiltonian)
we do not need to necessarily impose the relation $[X\times \bA^1, \pi_{\bA^1}]=0$. Thus,
we can consider the following variant of the Grothendieck ring with exponentials. 

\vspace{3pt}

\begin{definition}\label{D:6.1.2}

The coarse Grothendieck ring with exponentials $KExp^c(\cV_K)$ is
generated by isomorphism classes of pairs $(X,f)$ of a $K$--variety $X$ 
and a morphism $f: X\to \bA^1$ as above.

The relations etween them are generated by
$[X,f]=[Y,f|_Y] + [U,f|_U]$, 
for a closed subvariety $Y\hookrightarrow X$ and its open complement $U=X\setminus Y$.

The product is
$[X_1,f_1]\cdot [X_2,f_2]= [X_1\times X_2, f_1 \circ \pi_{X_1} + f_2 \circ \pi_{X_2}]$. 
\end{definition}

\vspace{3pt}

The Grothendieck ring with exponentials $KExp(\cV_K)$ is the quotient of this coarse 
version $KExp^c(\cV_K)$ by the ideal,  generated by $[\bA^1,id]$.

\vspace{3pt}

We will be using motivic measures and zeta functions that come from a 
choice of character $\chi$ as above, for which the elements $[X\times \bA^1, \pi_{\bA^1}]$ 
will be in the kernel. So all the motivic measures we will be
considering on $KExp^c(\cV_K)$ will factor through the Grothendieck
ring $KExp(\cV_K)$ of Definition \ref{D:6.1.1}

\vspace{3pt}

For a class $[X,f]\in KExp^c(\cV_K)$, the symmetric products are defined as
\[
[S^n(X,f)]:= [S^n(X),  f^{(n)}], 
\]
with $S^n(X)$ the symmetric product and 
$f^{(n)}: S^n(X) \to \bA^1$ given by
\[
  f^{(n)} [x_1, \ldots, x_n]=f(x_1)+\cdots+ f(x_n). 
\]
The analog of the Kapranov motivic zeta function in $KExp^c(\cV_K)$
is given by   
\[
Z_{(X,f)}(t)=\sum_{n\geq 0} [S^n(X, f)]\, t^n. 
\]
Given a motivic measure $\mu: KExp^c(\cV_K) \to R$, for some commutative ring $R$,
one can consider the corresponding zeta function
\[
 \zeta_\mu((X,f),t):= \sum_{n\geq 0} \mu(S^n X, f^{(n)}) \, t^n . 
\]

\vspace{3pt}

As our basic example, consider the finite field case $K=\cF_q$ and the motivic measures and zeta functions
discussed in Section~7.8 of \cite{MaMar21}, where the motivic measure $\mu_\chi : KExp^c(\cV_K) \to \bC$ 
is determined by a choice of character $\chi: \cF_q \to \bC^*$, 
\[
 \mu_\chi (X,f)=\sum_{x\in X(\cF_q)} \chi(f(x)), 
\]
and the associated zeta function is given by (see Proposition~ 7.8.1 of \cite{MaMar21})
\[
 \zeta_\chi((X,f),t)= \sum_n \sum_{\underline{x}\in S^n(X)(\cF_q)} \chi(f^{(n)}(\underline{x}))\, t^n =
\]
\[
 \exp\left( \sum_{m\geq 1} N_{\chi,m}(X,f) \, \frac{t^m}{m} \right),
 \]
with 
\[
 N_{\chi,m}(X,f)=\sum_\alpha  \sum_{r|m} r \, a_{\alpha,r} \, \alpha^{\frac{m}{r}}, 
\]
where for given $\alpha\in \bC$ and $r\in \N$, we have $a_{\alpha,r}=\mathrm{card}\,  X_{\alpha,r}$ for
the level sets
\[ 
X_{\alpha,r} :=\{ x\in X \, |\, [k(x): \bF_q] =r\, \text{ and } \chi(f(x))=\alpha \}. 
\]

\vspace{3pt}

We can then consider two possible variations, with respect to which we want to
compute a relative entropy through a Kullback--Leibler divergence: the variation
of the Hamiltonian, obtained through a change in the function $f: X\to \bA^1$, and
the variation in the choice of the character $\chi$ in this motivic measure $\mu_\chi$.
We will show how to simultaneously account for both effects.

\vspace{3pt}

To mimic the thermodynamic setting described in the previous subsection,
consider functions $F: X\times \bA^1 \to \bA^1$ of the form $F(x,\epsilon)=f_\epsilon(x)=f(x)+\epsilon \cdot h(x)$ for given 
morphisms $f,h: X \to \bA^1$, with $\epsilon\in \bG_m$ acting on $\bG_a=\bA^1$ by multiplication and $F(x,0)=f(x)$. 

\vspace{3pt}

Given a motivic measure $\mu_\chi: KExp^c(\cV)_{\cF_q} \to \bC$ associated to the choice of a character 
$\chi: \cF_q\to \bC^*$, and a class $[X,f]$ in $KExp^c(\cV)_{\cF_q}$, we consider the ``probability
distribution''
\[ 
P_{n,\underline{x}} = \frac{\chi(f^{(n)}(\underline{x}))\, t^n}{\zeta_\chi((X,f),t)}, 
\]
for $\underline{x}\in S^n(X)(\cF_q)$. We write $P_{n,\underline{x}}^{(f,\chi)}$ when
we need to emphasise the dependence on the morphism $f: X\to \bA^1$ and the
character $\chi$. We leave the $t$--dependence implicit. 
Of course, this is not a probability distribution in the usual sense, since it takes complex rather
than positive real values, though it still satisfies the normalisation condition. We still
treat it formally like a probability so that we can consider an associated notion of
Kullback--Leibler divergence $KL(P||Q)$.

\vspace{3pt}

Given a choice of a branch of the logarithm, to a character $\chi\in Hom(G,\bC^*)$, with
$G$ a locally compact abelian group, we can associate a group homomorphism
$\log \chi: G \to \bC$.

\vspace{3pt}

\begin{proposition} \label{P:6.1.3}

Given $[X,f]\in KExp^c(\cV)_{\bF_q}$, consider a class $[X\times \bA^1,F]$ with $F(x,0)=f$ and $F(x,\epsilon)=:f_\epsilon(x)$,
and characters $\chi,\chi': \cF_q\to \bC^*$. The Kullback--Leibler divergence then is
\[
KL(\zeta_\chi((X,f),t)||\zeta_\chi((X,f_\epsilon),t)):=  KL(P^{(f,\chi)}|| P^{(f_\epsilon,\chi)}) = 
\]
\[ 
\sum_{n,\underline{x}} P^{(f,\chi)}_{n,\underline{x}} \, \log \frac{P^{(f,\chi)}_{n,\underline{x}}}{P^{(f_\epsilon,\chi)}_{n,\underline{x}}} = \log \langle \chi_\epsilon(h) \rangle - \langle \log \chi_\epsilon(h) \rangle, 
\]
where $\langle \cdot \rangle$ is the expectation value with respect to $P^{(f,\chi)}$. 

Similarly, with
$\psi=\chi^{-1}\cdot \chi'$, we have
\[
 KL(\zeta_\chi((X,f),t)||\zeta_{\chi'}((X,f),t)) := KL(P^{(f,\chi)}|| P^{(f,\chi')}) = 
\]
\[ 
\sum_{n,\underline{x},\chi} P_{n,\underline{x},\chi} \log \frac{P_{n,\underline{x},\chi}}{P_{n,\underline{x},\chi'}} 
 = \log \langle \psi(f) \rangle - \langle \log \psi(f) \rangle. 
\]
\end{proposition}

\vspace{3pt}

\begin{proof}
We write the zeta function as
\[ 
\zeta_\chi((X,f_\epsilon),t)=\sum_n \sum_{\underline{x}\in S^n(X)(\cF_q)} \chi(f^{(n)}(\underline{x}))\, \chi(\epsilon
h^{(n)}(\underline{x}))\, t^n, 
\]
so that, if we formally regard as above
\[
P_{n,\underline{x}} = \frac{\chi(f^{(n)}(\underline{x}))\, t^n}{\zeta_\chi((X,f),t)} 
\]
as our ``probability distribution'', we have
\[
 \frac{ \zeta_\chi((X,f_\epsilon),t) }{\zeta_\chi((X,f),t)} = \sum_n \sum_{\underline{x}\in S^n(X)(\cF_q)} P_{n,\underline{x}} \,\,
\chi(\epsilon h^{(n)}(\underline{x})) = \langle \chi(\epsilon h) \rangle,
 \]
computing the expectation value $\langle \chi(\epsilon h) \rangle$ with respect to the distribution $P_{n,\underline{x}}$.

\vspace{3pt}

Then, setting as above
\[ 
P_{n,\underline{x},\epsilon} =P_{n,\underline{x}}^{(f_\epsilon,\chi)}= \frac{\chi(f^{(n)}(\underline{x}))\,\chi(\epsilon
h^{(n)}(\underline{x}))\, t^n}{\zeta_\chi((X,f_\epsilon),t)}, 
\]
we obtain
\[ 
\log P_{n,\underline{x},\epsilon} =\log P_{n,\underline{x}}  -\log \frac{\zeta_\chi((X,f_\epsilon),t)}{\zeta_\chi((X,f),t)}
+ \log \chi(\epsilon h^{(n)}(\underline{x})). 
\]
Thus, the Kullback--Leibler divergence gives
\[
KL(\zeta_\chi((X,f),t)||\zeta_\chi((X,f_\epsilon),t)):= 
\sum_{n,\underline{x}} P_{n,\underline{x}} \log \frac{P_{n,\underline{x}}}{P_{n,\underline{x},\epsilon}} =
\]
\[
 \log \frac{\zeta_\chi((X,f_\epsilon),t)}{\zeta_\chi((X,f),t)} - \sum_{n,\underline{x}} P_{n,\underline{x}}\,\, \log \chi_\epsilon( h^{(n)}(\underline{x})), 
\]
where $\chi_\epsilon$ is the character $\chi_\epsilon =\chi \circ \epsilon$. We can write this equivalently as
\[ 
KL(\zeta_\chi((X,f),t)||\zeta_\chi((X,f_\epsilon),t)) = \log \langle \chi_\epsilon(h) \rangle - \langle \log \chi_\epsilon(h) \rangle. 
\]

\vspace{3pt}

Similarly, given two characters $\chi,\chi': \bF_q \to \bC^*$, let $\psi=\chi^{-1}\cdot \chi'$, the Kullback--Leibler 
divergence is given by 
\[ 
KL(\zeta_\chi((X,f),t)||\zeta_{\chi'}((X,f),t)) = \sum_{n,\underline{x},\chi} P_{n,\underline{x},\chi} \log \frac{P_{n,\underline{x},\chi}}{P_{n,\underline{x},\chi'}} , 
\]
where $\log P_{n,\underline{x},\chi'} = \log P_{n,\underline{x},\chi} + \log \psi(f^{(n)}(\underline{x}))$, so that
\[ \begin{aligned}
KL(\zeta_\chi((X,f),t)||\zeta_{\chi'}((X,f),t)) &= \log \frac{\zeta_{\chi'}((X,f),t)}{\zeta_\chi((X,f),t)} -  \sum_{n,\underline{x}} P_{n,\underline{x}}\,\, \log \psi( f^{(n)}(\underline{x}))\\
&=\log \langle \psi(f) \rangle - \langle \log \psi(f) \rangle,
\end{aligned}\] 
as stated. 
\end{proof}

\vspace{3pt}

The notion of Kullback--Leibler divergence for zeta functions considered above requires that
the comparison is made at the same class $[X,f]$ in $KExp^c(\cV_K)$. If one wants to
introduce the possibility of comparing the zeta functions at two different classes $[X_1,f_1]$
and $[X_2,f_1]$ via a Kullback--Leibler divergence, it makes sense to compare them over
the fibered product, namely the natural space where the morphisms $f_1,f_2$ agree. 

Namely
we define $KL(\zeta_\chi((X_1,f_1),t)||\zeta_\chi((X_2,f_{2,\epsilon}),t))$ to be given by
\[
KL(\zeta_\chi(X_1\times_{f_1,f_2} X_2, f),t)||\zeta_\chi(X_1\times_{f_1,f_2} X_2, f_\epsilon),t)), 
\]
using the pullback $X_1\times_{f_1,f_2} X_2=\{ (x_1,x_2)\in X_1\times X_2\,|\, f_1(x_1)=f_2(x_2)\}$ with
$f=f_1\circ\pi_1=f_2\circ \pi_2: X_1\times_{f_1,f_2} X_2 \to \bA^1$ and $f_\epsilon=F(\cdot,\epsilon)$ for
some $F: X_1\times_{f_1,f_2} X_2 \times \bA^1 \to \bA^1$ with $F(\cdot,0)=f$.

\vspace{3pt}

\subsection{Shannon entropy and Hasse--Weil zeta function} \label{S:6.2}

A notion of Shannon information in the context of Grothendieck rings of varieties was
proposed in \cite{Mar19b}, based on regarding zeta functions as physical partition
functions. We show here that this can be regarded as a special case of the construction
described above.

\smallskip

First observe that the usual Grothendieck ring of $K$-varieties $K_0(\cV_K)$
embeds in the coarse Grothendieck ring with exponentials $KExp^c(\cV_K)$
by mapping $[X]$ to $[X,0]$. For varieties over a finite field $K=\bF_q$, 
the measure $\mu_\chi$ associated to the trivial character $\chi=1$ is just the 
counting measure $\mu_1(X,f)=\mathrm{card}\,  X(\bF_q)$, hence the zeta function $\zeta_{\mu_1}$
restricted to $K_0(\cV_K) \subset KExp^c(\cV_K)$ is the Hasse-Weil zeta function of
$X$, 
\[
Z^{HW}(X,t)=\exp\left( \sum_m \mathrm{card}\, X(\bF_{q^m}) \, \frac{t^m}{m} \right).
\]
This can be seen by writing $Z^{HW}(X,t)$ in terms of effective zero--cycles as
\[
 Z^{HW}(X,t)=\sum_\alpha t^{\deg\alpha},
\]
and further writing the latter in the form
\[
 Z^{HW}(X,t)=\sum_{n\geq 0} \mathrm{card}\,  S^n(X)(\bF_q) \, t^n . 
 \]

\vspace{3pt}

In the above expression, one can regard the quantity
\[
P(\alpha):=\frac{t^{\deg\alpha}}{Z^{HW}(X,t)} 
\]
as a probability measure assigned to the zero--cycle $\alpha$, hence one can consider
the Shannon information of this distribution, which is given by
\[ 
S(X,t):=-\sum_\alpha P(\alpha)\log P(\alpha) = \log Z^{HW}(X,t) + Z^{HW}(X,t)^{-1} H(X,t), 
\]
\[ 
H(X,t):= - \sum_\alpha t^{\deg\alpha} \log (t^{\deg\alpha}). 
\]

\vspace{3pt}

To compare this to the Kullback--Leibler divergence introduced in the previous
section, one can equivalently regard the above expression as the Shannon entropy
of the distribution $P_{n,\underline{x}}=P_{n,\underline{x}}^{(f=0,\chi=1)}$,
\[
 P_{n,\underline{x}} = \frac{t^n}{Z^{HW}(X,t)},
 \]
for all $n\geq 0$ and all $\underline{x}\in S^n(X)(\bF_q)$,
\[
 S(X,t)=-\sum_{n,\underline{x}} P_{n,\underline{x}}  \log P_{n,\underline{x}} . 
 \]

\vspace{3pt}

It is customary to make a change of variables $t=q^{-s}$ and write the Hasse--Weil zeta
function as $Z^{HW}(X,q^{-s})$. We correspondingly write $S(X,s)$ for the Shannon
entropy defines as above, after setting $t=q^{-s}$.

\vspace{3pt}

\begin{lemma}\label{L:6.2.1}

 For a variety $X$ over $\bF_q$ the Shannon entropy associated to the 
Hasse--Weil zeta function is given by
\[ 
S(X,s)= \left( 1- s\frac{d}{ds} \right) Z^{HW}(X,q^{-s}). 
\] 
\end{lemma}

\vspace{3pt}

\begin{proof}
With the associated probability distribution 
$P_{n,\underline{x}} = \frac{q^{-sn}}{Z^{HW}(X,q^{-s})}$ we have
\[\begin{aligned}
S(X,s)&= - \sum_{n,\underline{x}}  \frac{q^{-sn}}{Z^{HW}(X,q^{-s})} (\log q^{-sn} -\log Z^{HW}(X,q^{-s}))\\
& = \log Z^{HW}(X,q^{-s}) + Z^{HW}(X,q^{-s})^{-1}  \sum_n \mathrm{card}\,  S^n(X)(\bF_q) \, q^{-sn}\, sn \log q.  
\end{aligned}\]
\end{proof}
\medskip

\subsubsection{Thermodynamical interpretation of the Shannon entropy} \label{S:6.2.1}

\ 

Lemma~\ref{L:6.2.1} shows that the Shannon entropy $S(X,s)$ agrees with
the usual thermodynamical entropy
\[ 
S= \left(1-\beta \frac{\partial}{\partial\beta}\right) \log Z(\beta) 
\]
of a physical system with partition function $Z(\beta)$ at inverse temperature $\beta>0$,
and free energy $F=-\log Z(\beta)$. 
We identify here the Hasse--Weil zeta function $Z^{HW}(X,q^{-s})$ with the partition
function of a physical system with Hamiltonian $H$ with energy levels 
$Spec(H)=\{ n \log q \}_{n\geq 0}$
with degeneracies  $\mathrm{card}\,  S^n(X)(\bF_q)$, so that
\[
 \mathrm{Tr}(e^{-\beta H}) = Z^{HW}(X,q^{-\beta}). 
\]
The expression given above for the thermodynamical entropy is the same as the
Shannon entropy of the probability distribution $P_n =\frac{e^{-\beta \lambda_n}}{Z(\beta)}$,
for $Z(\beta)=\mathrm{Tr}(e^{-\beta H})$ with $Spec(H)=\{ \lambda_n \}$, since we have
\[ 
S=-\sum_n P_n \log P_n = \sum_n P_n \log Z(\beta) +\beta \sum_n P_n  \lambda_n,
\]
where $\sum_n P_n  \lambda_n =  \frac{\partial}{\partial\beta} \log Z(\beta)$.

\medskip

\begin{lemma} \label{L:6.2.3}

 The entropy of $Spec(\bF_q)$, 
\[
S(Spec(\cF_q),s) =- \left( 1- s\frac{d}{ds} \right) \log(1-q^{-s}), 
\]
is the thermodynamical entropy of a physical system with non-degenerate energy levels 
$\{ k \log q \}_{k\geq 0}$.
The entropy of an affine space $\bA^n$ over $\bF_q$,
\[
S(\bA^n_{\bF_q},s) = -\left( 1- s\frac{d}{ds} \right) \log(1-q^{-s+n}), 
\]
is the thermodynamical entropy of a physical system with energy levels $\{ k \log q \}_{k\geq 0}$
with degeneracies $q^{kn}$.
The entropy of a projective space $\bP^n$ over $\bF_q$,
\[
S(\bP^n_{\bF_q},s) = -\left( 1- s\frac{d}{ds} \right) \sum_{k=0}^n \log(1-q^{-s+k}), 
\]
is the thermodynamical entropy of a composite system consisting of $n+1$ independent
subsystems, all of them with energy levels $\{ k \log q \}_{k\geq 0}$, where the $j$--th
system has energy levels with degeneracies $q^{kj}$.
\end{lemma}

 \vspace{3pt}

 \begin{proof}
This follows directly from Lemma\, \ref{L:6.2.1}, since the Hasse--Weil zeta functions are respectively given by 
\[ \begin{aligned} 
Z^{HW}(Spec(\bF_q),q^{-s})  &=\exp\,(\sum_m  \frac{q^{-sm}}{m})=\exp(-\log(1-q^{-s})) \\
& = \frac{1}{1-q^{-s}}\\
&=\sum_{k\geq 0} q^{-s k}.
\end{aligned}\]
\[ \begin{aligned} 
 Z^{HW}(\bA^n,q^{-s}) &= \exp(\sum_m \frac{q^{mn} q^{-sm}}{m})\\
 &=\exp(-\log(1-q^n q^{-s}))\\
 & =\frac{1}{1-q^{-s+n}}\\
&=\sum_{k\geq 0} q^{kn}\, q^{-s k}.
\end{aligned}\]
\[ 
Z^{HW}(\bP^n,q^{-s}) = Z^{HW}(\bA^0\cup \bA^1\cup \cdots \cup \bA^n,q^{-s})=\frac{1}{(1-q^{-s})\cdots (1-q^{-s+n})}. 
\]
The case of $\bP^n$ reflects the usual property of additivity of the Shannon entropy over independent subsystems.
\end{proof}

\vspace{3pt}

The following equivalent description of the Shannon entropy of $\bP^n$ over $\bF_p$
will become useful in the next subsection.

\vspace{3pt}

A matrix $M=(M_{ij})\in M_{n\times n}(\bZ)$ is {\it reduced} if it is a lower triangular with
$0\leq M_{ij}\leq M_{jj}$ for $i\geq j$. Let $Red_n$ denote the set of reduced matrices.
For a given positive integer $m\in \N$ let $Red_n(m)=\{ M\in Red_n \,|\, \det(M)=m \}$.

\vspace{3pt}

\begin{lemma} \label{L:6.2.4}

 The entropy $S(\bP^{n-1}_{\bF_p},s)$ can also be identified
with the thermodynamical entropy of a physical system with energy levels $\{ k \log p \}_{k\geq 0}$
and degeneracies $D_k=\mathrm{card}\, Red_n(p^k)$.
\end{lemma}

\vspace{3pt}

\begin{proof}
If a reduced matrix $M$ has $\det(M)=p^k$ then the $j$-th diagonal entry is $p^{k_j}$
with $\sum_j k_j =k$. For a given diagonal entry $p^{k_j}$ there is a total of $p^{k_j (k-j)}$ possibilities
in the $j$-th column satisfying $0\leq M_{ij}\leq M_{jj}$ for $i\geq j$.
Thus we can write the multiplicities as $D_k=\sum_{k_1+\cdots+k_n=k} D_{k_1,\ldots, k_n}$
with $D_{k_1,\ldots,k_n}=p^{k_1 (n-1)} p^{k_2 (n-2)} \cdots p^{k_{n-1}}$. The partition function
of such a system is given by
\[ \begin{aligned} 
\mathrm{Tr}(e^{-sH})&= \sum_{k\geq 0}  p^{-s k}  \sum_{k_1+\cdots+k_n=k} D_{k_1,\ldots, k_n}\\ 
&= \sum_{k\geq 0} p^{-s k}\, \sum_{k_1+\cdots+k_n=k} p^{k_1 (n-1)} p^{k_2 (n-2)} \cdots p^{k_{n-1}}\\
&= \sum_{k_1,\ldots, k_n} p^{k_1 (n-1-s)} p^{k_2 (n-2-s)} \cdots p^{k_{n-1} (1-s)} p^{- k_n s}\\ 
&=\prod_{\ell=0}^{n-1} \frac{1}{1-p^{-s+\ell}}\\
& = Z^{HW}(\bP^{n-1}_{\bF_p},p^{-s}). 
\end{aligned}\]

This identifies the thermodynamical entropy of this system with the Shannon entropy $S(\bP^{n-1}_{\bF_p},s)$.
\end{proof}

\vspace{3pt}

\subsection{The case of varieties over $\bZ$}\label{S:6.3}
Consider a variety $X$ over $\bZ$ and denote by $X_p$ 
 the reduction of $X$ mod $p$. 
For simplicity, we will consider here only the case where 
there are no primes of bad reduction. The
(non-completed) $L$-function is then given by
\[
L(X,s)=\prod_p Z^{WH}(X_p, p^{-s}), 
\]
while the completed $L$--function includes a contribution of the archimedean prime ([Se70]):
\[
L^*(X,s)=L(X,s) \cdot L_\infty (X,s), 
\]
\[ 
L_\infty (X,s) := \prod_{i=0}^{\dim X} L_\infty(H^i(X),s)^{(-1)^{i+1}}, 
\]
\[
L_\infty(H^i(X),s) :=  \prod_{p<q} \Gamma_\bC(s-p)^{h^{p,q}} \prod_p \Gamma_\bR(s-p)^{h^{p,+}} \Gamma_\bR (s-p+1)^{h^{p,-}},   
\]
where $h^{p,q}$ are the Hodge numbers of the complex variety $X_\bC$, with $h^{p,\pm}$ the dimension of
the $(-1)^{p}$--eigenspace of the involution on $H^{p,p}$ induced by the real structure, and
\[
 \Gamma_\bC(s):=(2\pi)^{-s} \Gamma(s), \ \ \ \  \Gamma_\bR(s)=2^{-1/2} \pi^{-s/2} \Gamma(s/2). 
\]
In the more general cases of number fields with several archimedean places, corresponding to
the embeddings of the number file in $\bC$, the archimedean places given by real embeddings
have an archimedean factor as above.  The archimedean places given by complex embeddings 
have a similar one, that also depends on the Hodge structure, of the form
\[ 
\prod_{p,q}\Gamma_\bC (s-\min(p,q))^{h^{p,q}}. 
\]
As shown in \cite{Se70}, the form of these archimedean local factors is 
dictated by the expected form of the conjectural functional equation for the completed $L$--function
$L^*(X,s)$. 

\vspace{3pt}

The non--completed $L$--function $L(X,s)$ can be understood as the partition function of
a physical system consisting of a countable family of independent subsystems, one for each
prime, with partition function $Z^{WH}(X_p, p^{-s})$. Hence, the additivity of the Shannon
entropy over independent subsystems prescribes that the associated entropy is
\[
S_\bZ(X,s):=\sum_p S(X_p,s) =H_\bZ(X,s)+\log L(X,s), 
\]
\[ 
H_\bZ(X,s):= \sum_p Z^{HW}(X_p,p^{-s})^{-1} H(X_p,p^{-s}). 
\]
This can be equivalently written as
\[ 
S_\bZ(X,s) = \left(1-s\frac{d}{ds}\right) \log L(X,s). 
\]

\vspace{3pt}

The following statement is  a direct consequence of Lemma~\ref{L:6.2.1}, Lemma~\ref{L:6.2.3} and
the above definition of the entropy $S_\bZ(X,s)$. 

\vspace{3pt}

\begin{lemma} \label{L:6.3.1}

 The Shannon entropy of the non-completed $L$-function for $Spec(\bZ)$,
\[ 
S_\bZ(Spec(\bZ),s)= \left(1-s\frac{d}{ds}\right)\log \zeta(s) 
\]
is the thermodynamical entropy of a physical system with non--degenerate energy 
levels $\{ \log k \}_{k \geq 1}$. 

Such physical systems are realised, for instance, by
the Julia system of \cite{Ju90} {\it or by the Bost--Connes system of} \cite{BoCo95}. 

{\it The Shannon
entropy of the non--completed $L$--function for an affine space $\bA^n$ over $\bZ$,
\[
S_\bZ(\bA^n,s)=\left(1-s\frac{d}{ds}\right)\log \zeta(s-n), 
\]
is the thermodynamical entropy of a physical system with energy 
levels $\{ \log k \}_{k \geq 1}$ with degeneracies $k^n$. The Shannon
entropy of the non--completed $L$--function for a projective space $\bP^n$ over $\bZ$,
\[ 
S_\bZ(\bP^n,s)=\left(1-s\frac{d}{ds}\right)\sum_{m=0}^n \log \zeta(s-m), 
\]
is the thermodynamical entropy of a physical system with energy levels 
$\{ \log k \}_{k \geq 1}$ and degeneracies $D_k=\mathrm{card}\, Red_n(k)$.

\vspace{3pt}

Such physical systems are realised by the $GL_n$--versions of the Bost--Connes system
considered in} \cite{Sh16}.
\end{lemma}

\vspace{3pt}

When one considers also the archimedean places, one can regard the 
completed partition function $L^*(X,s)$ in a similar way as the partition
function of a composite system consisting of independent subsystems
for each non-archimedean and archimedean prime. The contributions
of the non-archimedean primes are given, as above, by systems with
partition function given by the Hasse-Weil zeta function $Z^{HW}(X_p,p^{-s})$,
while the contribution of the archimedean places has partition function
$L_\infty(X,s)$. The additivity of the Shannon entropy over independent
subsystems again prescribes that we assign entropy
\[
S^*_\bZ(X,s):=\sum_p S(X_p,s)  + S_\infty(X,s), 
\]
where $S_\infty(X,s)$ is the entropy of a system with partition function
$L_\infty(X,s)$,
\[ 
S_\infty(X,s):= \left(1-s\frac{d}{ds}\right) \log L_\infty(X,s). 
\]

\smallskip

The difficulty here is in interpreting the expression
\[ 
L_\infty(H^i(X),s) :=  \prod_{p<q} \Gamma_\bC(s-p)^{h^{p,q}} \prod_p \Gamma_\bR(s-p)^{h^{p,+}} \Gamma_\bR (s-p+1)^{h^{p,-}}   
\]
as a partition function (see also \cite{Ju90}).
This problem is closely related to the well known
arithmetic problem of obtaining an interpretation of the archimedean factors that 
parallels the corresponding form of the non-archimedean ones, see 
\cite{Den91}, \cite{Ma95}. 
Note that the logarithmic derivative of these non-archimedean local factors, which
determines the associated entropy, has a Lefschetz trace formula interpretation as
proved in Sec.~7 of \cite{CCM07}. 

\vspace{3pt}

\subsection{Shannon entropy for other motivic measures}\label{S:6.4}

Denote by $Z^{mot}(X,t)$  the Kapranov motivic zeta function (see \cite{Kap00}):
\[ 
Z^{mot}(X,t)=\sum_{n=0}^\infty [S^n X]\, t^n .
\]
It is a formal power series in $K_0(\cV)[[t]]$. 

\vspace{3pt}

Starting with a motivic measure, given by a ring homomorphim
$\mu: K_0(\cV)\to R$ with values in a commutative ring $R$, 
we can define the zeta function $\zeta_\mu$ by applying the 
motivic measure $\mu$ to the coefficients of the Kapranov zeta function:
\[ 
\zeta_\mu(X,t)=\sum_{n=0}^\infty \mu(S^n X) t^n . 
\]
In particular, a motivic measure $\mu: K_0(\cV)\to R$ is called {\it exponentiable}, if
the zeta function $\zeta_\mu$ defines a {\it ring} homomorphism 
$\zeta_\mu: K_0(\cV)\to W(R)$ to the Witt ring of $R$: see \cite{Ram15}, \cite{RamTab15}. 

\vspace{3pt}

For any $\mu$, the zeta function 
$\zeta_\mu$ defines an additive map, which means that it satisfies the inclusion--exclusion property
\[
\zeta_\mu (X\cup Y, t)=\frac{\zeta_\mu(X,t) \zeta_\mu(Y,t)}{\zeta_\mu(X\cap Y,t)}, 
\]
where the product of power series is the {\it addition} in the Witt ring.

\vspace{3pt}

The exponentiability
 means that one also has
\[ 
\zeta_\mu (X\times Y, t)=\zeta_\mu(X,t) \star \zeta_\mu(Y,t), 
\]
where $\star$ is the {\it product} in the Witt ring, which is uniquely determined by 
\[ 
(1-at)^{-1}\star (1-bt)^{-1}=(1-ab t)^{-1}, \ \ \  a.b \in R. 
\]
The motivic measure given by  counting points over a finite field, with $\zeta_\mu$ the Hasse--Weil
zeta function, is exponentiable. 

\vspace{3pt}

Given a motivic measure $\mu: K_0(\cV)\to R$ 
with associated zeta function
\[ 
\zeta_\mu(X,t)=\sum_{n=0}^\infty \mu(S^n X) t^n, 
\]
one can define the respective Shannon entropy as
\[ 
S_\mu(X,t):=(1-t \log t \frac{d}{dt}) \log \zeta_\mu(X,t), 
\]
where the $\log t$ factor occurs due to a change of variables $t=\lambda^{-s}$
with respect to the thermodynamic entropy with inverse temperature $\beta=\log\lambda$. 

\vspace{3pt}

As discussed in [Mar19b], the entropy $S_\mu$ satisfies an analog of the Khinchin
axioms for the usual Shannon entropy, where the extensivity property of the Shannon
entropy on composite system is expressed as the inclusion--exclusion property
\[
S_\mu(X\cup Y,t)=S_\mu(X,t)+S_\mu(Y,t)-S_\mu(X\cap Y,t). 
\]

\vspace{3pt}

As we discussed at the beginning of Section~\ref{S:6.2.1}, 
this definition of entropy is justified by interpreting the zeta function $\zeta_\mu(X,t)$
as a partition function and its logarithm $-\log \zeta_\mu(X,t)$ as the associated free energy.
In the following subsection we return to the relative entropy (Kullback--Leibler 
divergence) introduced in Section~6.1, and we discuss the corresponding thermodynamical
interpretation.

\vspace{3pt}

\subsection{Thermodynamics of Kullback--Leibler divergence} \label{S:6.5}

We recalled above the thermodynamical entropy
\[ 
S= \left(1-\beta \frac{\partial}{\partial\beta}\right) \log Z(\beta) 
\]
of a physical system with partition function given by the zeta function 
$Z(\beta)$ at inverse temperature $\beta>0$ with free energy $F=-\log Z(\beta)$.

\smallskip

The Kullback--Leibler divergence
\[ 
KL_\phi(P|| Q)=\sum_{x\in X} P_x \log \frac{P_x}{Q_x} 
\]
of two probabilities $P,Q$ on the same set $X$ can also be interpreted in terms
of free energy and Gibbs free energy. If $Q_x=\frac{e^{-\beta H_x}}{Z(\beta)}$ with $Z(\beta)=\sum_x e^{-\beta H_x}$ the
partition function, and $P$ a given probability distribution on the configuration space,
the Gibbs free energy is given by
\[
G(P) =-\log Z(\beta) + \sum_x P_x \log \frac{P_x}{Q_x}, 
\]
hence $KL(P||Q)=G(P)+\log Z(\beta)$. It follows that the usual free energy is minimization
of the Gibbs energy over configuration probabilities: since $KL(P||Q)\geq 0$, we have
\[
 \min_P G(P)= -\log Z(\beta). 
\]

\vspace{3pt}

In the typical approach of mean field theory, when computation of the free
energy of a system is not directly accessible, one considers a trial Hamiltonian $\tilde H$
with probability distribution $P_x =\tilde Z(\beta)^{-1} e^{-\beta \tilde H_x}$,
where $\tilde Z(\beta) :=\sum_x e^{-\beta \tilde H_x}$. 
Then $\log P_x=-\log \tilde Z(\beta) -\beta \tilde H_x$, and the Helmholtz
free energy 
\[
 - \sum_x P_x \log P_x = \log \tilde Z(\beta) +\beta \langle \tilde H \rangle 
= (1-\beta \frac{\partial}{\partial \beta}) \log \tilde Z(\beta)
\]
satisfies
\[
\sum_x P_x \log \frac{P_x}{Q_x} = \log\frac{Z(\beta)}{\tilde Z(\beta)} + \beta \langle H - \tilde H \rangle. 
\]

In the mean field theory setting, it is common to also assume that the
trial Hamiltonian satisfies $\langle H \rangle =\langle \tilde H \rangle$ with the averages
computed with respect to the probability distribution $P_x$, so that the identity above would reduce just to
\[
\sum_x P_x \log \frac{P_x}{Q_x} = -\log \tilde Z(\beta) +\beta \langle \tilde H \rangle 
+\log Z(\beta) -\beta \langle H \rangle = \log\frac{Z(\beta)}{\tilde Z(\beta)}. 
\]

However, we will not make this assumption.

\vspace{3pt}

If we consider the case of a $1$--parameter family of {\it commuting} Hamiltonians 
$H(\epsilon)$ that depends analytically on $\epsilon$, such that 
\[
H(\epsilon) = \tilde H + \epsilon \frac{\partial \tilde H}{\partial \epsilon}|_{\epsilon=0} + O(\epsilon^2), 
\]
we have 
\[
\sum_x P_x H_x(\epsilon) \sim \sum_x P_x \tilde H_x + \epsilon \sum_x 
P_x \frac{\partial H_x(\epsilon)}{\partial \epsilon}|_{\epsilon =0}. 
\]
The generalized force corresponding to the variable $\epsilon$
is given by \[L_x=- \frac{\partial H_x(\epsilon)}{\partial \epsilon}|_{\epsilon =0}. \]
Then
\[
\langle L \rangle= \sum_x P_x L_x =\frac{1}{\beta} \frac{\partial}{\partial \epsilon} \log Z_\epsilon(\beta) |_{\epsilon=0}, 
\]
where $Z_\epsilon(\beta) = \sum_x e^{-\beta H(\epsilon)}$. For $P_x(\epsilon)=Z_\epsilon(\beta)^{-1} e^{-\beta H(\epsilon)},$
we have
\[
 \log P_x(\epsilon) = -\log Z_\epsilon(\beta) - \beta ( \tilde H_x + \epsilon L_x +O(\epsilon^2) ). 
\]
Thus, we can write the Kullback--Leibler divergence in this case as
 \[ \begin{aligned}
 \sum_x P_x &\log \frac{P_x}{P_x(\epsilon)} =\sum_x P_x \log P_x + \log Z_\epsilon(\beta) 
+\beta  \sum_x P_x \tilde H_x
+\epsilon \beta\sum_x P_x L_x
+O(\epsilon^2)\\
& =-(1-\beta \frac{\partial}{\partial \beta}) \log \tilde Z(\beta) + \log Z_\epsilon(\beta) 
+\beta  \sum_x P_x \tilde H_x
+\epsilon \frac{\partial}{\partial \epsilon} \log Z_\epsilon(\beta)|_{\epsilon=0} + O(\epsilon^2)\\
& =\log \frac{Z_\epsilon(\beta)}{\tilde Z(\beta)} +\epsilon \frac{\partial}{\partial \epsilon} \log Z_\epsilon(\beta) |_{\epsilon=0} 
+ O(\epsilon^2). 
\end{aligned}\]

\vspace{3pt}

This thermodynamical point of view on the Kullback--Leibler divergence has the advantage that one can
express its leading term purely in terms of partition functions, suggesting how to develop 
a motivic analog defined in terms of zeta and $L$--functions. This motivates the definition of
Kullback--Leibler divergence that we used in Proposition~\ref{P:6.1.3}.

\vspace{3pt}

\subsection{Fisher--Rao metric}\label{S:6.6}

Recall that, for a family of probability distributions $P(\gamma)=(P_x(\gamma))$ on a set $X$,
depending differentiably on parameters $\gamma=(\gamma_1,\ldots,\gamma_r)$, the Fisher--Rao 
information metric is defined as 
\[
 g_{ij}(\gamma):= \sum_x P_x(\gamma) \, \frac{\partial \log P_x(\gamma)}{\partial \gamma_i} \frac{\partial \log P_x(\gamma)}{\partial \gamma_j} . 
\]
Assuming that the probability distributions $P_x(\gamma)$ are of the form
\[
P_x(\gamma)= \frac{e^{-\beta H_x(\gamma)}}{Z_\gamma(\beta)}, \ \ \  
Z_\gamma(\beta)=\sum_x e^{-\beta H_x(\gamma)}, 
\]
for a family of commuting Hamiltonians $H(\gamma)$, we can write as above
\[
 \log P_x(\gamma)= -\log Z_\gamma(\beta) - \beta H_x(\gamma). 
\]
If we consider the generalized forces 
\[
L_{x,i}= - \frac{\partial H_x(\gamma)}{\partial \gamma_i}, 
\]
we can write the Fisher-Rao metric as
\[ 
g_{ij}(\gamma)= \frac{\partial \log Z_\gamma(\beta)}{\partial \gamma_i} \frac{\partial \log Z_\gamma(\beta)}{\partial \gamma_j} + \beta^2 \sum_x P_x(\gamma) L_{x,i} L_{x,j}. 
\]

\vspace{3pt}

Equivalently, the Fisher--Rao metric can be obtained as the Hessian matrix of the Kullback--Leibler divergence 
\[ 
g_{ij}(\gamma_0)=\frac{\partial^2}{\partial \gamma_i \partial \gamma_j} \, KL(P(\gamma)||P(\gamma_0)) |_{\gamma =\gamma_0}. 
\]

\vspace{3pt}

Another way of characterizing the Fisher--Rao metric tensor in the case of classical information theory
is designed for easier comparison with the case of quantum information, and will also be more
directly useful in our setting. 

\vspace{3pt}

One identifies classical probability distributions with diagonal density matrices.
This is equivalent to considering
pairs $\rho,\rho'\in M_N$ of commuting density matrices, $[\rho,\rho']=0$, so that
\[
KL(\rho || \rho')=Tr (\rho\, (\log \rho - \log \rho')) =\sum_i P_i \log \frac{P_i}{Q_i}, 
\]
where $\rho$ and $\rho'$ are simultaneously diagonalized to $\rho=diag(P_i)$ and $\rho'=diag(Q_i)$.
Since we have $KL(\rho || \rho')=Tr (\rho\, (\log \rho - \log \rho'))=\infty$ whenever $0\in Spec(\rho')$
we can restrict to the case where $\rho'$ is invertible. 

\vspace{3pt}

Let $h$ be an (infinitesimal) increment, so that $\rho + h \in \cM^{(N)}$. In particular, 
this implies $Tr(h)=0$. We are interested in computing the relative entropy
$KL(\rho + h || \rho)$, up to second order terms in $h$. Again, we assume $\rho$ is
invertible, to ensure $KL(\rho + h || \rho)< \infty$.

\vspace{3pt}

\begin{lemma}\label{L:6.6.1}

In the classical case, or equivalently whenever $[\rho, h]=0$, we simply have
\[
 KL(\rho + h || \rho) = \frac{1}{2} Tr(h \rho^{-1} h) + O(h^3),
\] 
where the right hand side is the classical Fisher metric. 
\end{lemma}

\vspace{3pt}

\begin{proof}
We have
\[\begin{aligned}
KL(\rho + h || \rho)  &=  Tr( (\rho+h) \log(\rho +h) ) - Tr( (\rho +h) \log \rho) \\
&=  Tr (\rho \log(\rho (I +\rho^{-1} h))) +  Tr(h \log(\rho (I + \rho^{-1} h)))\\
&\quad   - Tr(\rho \log \rho) - Tr (h \log\rho) \\
&=   Tr(\rho\log (I +\rho^{-1} h))  + Tr(h \log (I +\rho^{-1} h)). 
\end{aligned}\]
Expanding up to second order 
$\log (I +\rho^{-1} h)=\rho^{-1} h - \frac{1}{2} \rho^{-1} h \rho^{-1} h +O(h^3)$,
we obtain
\[ 
KL(\rho + h || \rho) = Tr(\rho\, \rho^{-1} h ) -  \frac{1}{2}Tr(\rho\,  \rho^{-1} h \rho^{-1} h)
+ Tr(h \rho^{-1} h) + O(h^3) = \frac{1}{2} Tr(h \rho^{-1} h), 
\]
where the condition $Tr(h)=0$ gives the vanishing of the first order term. 
\end{proof}

\vspace{3pt}

\subsection{Fisher--Rao tensor and zeta functions}\label{S:6.7}

We now return to our setting in the coarse Grothendieck ring of varieties with exponentials $KExp^c(\cV)_K$
and we discuss how to obtain a Fisher--Rao tensor based on the notion of Kullback--Leibler divergence
discuss above. 

\vspace{3pt}

Given a  variety $X$ over a field $K$, we define as in [DenLoe01] the arc space 
$\cL(X)$ of $X$ as the projective limit of the truncated arc spaces $\cL_m(X)$, where
$\cL_m(X)$ is the variety over $K$ whose $L$-rational points, for a field $L$ containing $K$,
are the $L[u]/u^{m+1}$-rational points of $X$. In particular, $\cL_0(X)=X$ and $\cL_1(X)$
is the tangent bundle of $X$. Consider a morphism $f: X \to \bA^1$ and the induced 
morphisms $f_m: \cL_m(X) \to \cL_m(\bA^1)$. Points in $\cL(\bA^1)$ can be identified with
power series $\alpha(u)$ in $L[[u]]$, for some field extension $L$, or respectively in $L[[u]]/u^{m+1}$
in the case of $\cL_m(\bA^1)$. 

\vspace{3pt}

Given a pair $(X,f)$ with $f: X \to \bA^1$, we can consider the associated pairs $(\cL_m(X),f_m)$
with the induced morphism $f_m: \cL_m(X)\to \cL_m(\bA^1)$ and similarly for the arc space
$\cL(X)$. For $K=\bF_q$, consider characters $\chi_m$ given by group homomorphisms 
\[
\chi_m: \bF_q [[u]]/u^{m+1} \to \bC^*, 
\]
so that we can compute
\[ 
\mu_{\chi_m}(\cL_m(X),f_m)=\sum_{\varphi \in \cL_m(X)(\bF_q)} \chi(f_m(\varphi)) .
 \]
The respective  zeta function will be of the form
\[ \begin{aligned}
 \zeta_{\chi_m}((\cL_m(X),f_m),t) &=\sum_{n\geq 0} \mu_{\chi_m} (S^n((\cL_m(X),f_m)))\, t^n 
\\
 &=\sum_n \sum_{\underline{\varphi}\in S^n(\cL_m(X))(\bF_q)} \chi_m(f^{(n)}_m(\underline{\varphi})) \, t^n, 
\end{aligned}\]
where the symmetric products are given by {\small $S^n((\cL_m(X),f_m))=(\cL_m(S^n(X)),f^{(n)}_m)$}
with the morphism $f^{(n)}_m:\cL_m(S^n(X)) \to \cL_m(\bA^1)$ induced  by
$f^{(n)}: S^n(X)\to \bA^1$. 

\vspace{3pt}

For the purpose of constructing a Fisher-Rao tensor, it suffices to work with the tangent
bundle, namely with the first truncated arc space $\cL_1(X)$. 

\medskip

\begin{lemma} \label{L:6.7.1}

 Let $K=\bF_q$ be a finite field. Consider a class 
$[ \cL_1(X), f_1 ]\in KExp^c(\cV)_K$ 
and a character $\chi_1: K[u]/u^2 \to \bC^*$. Then, for all $n\in \N$,  the summands
$h_{n,\underline{\varphi}} := \chi_1(f^{(n)}_1(\underline{\varphi})) t^n$
satisfy $\sum_{\underline{\varphi}} h_{n,\underline{\varphi}} =0$.
\end{lemma}

\vspace{3pt}

\begin{proof}
Elements $\varphi \in \cL_1(X)(K)$ can be interpreted as specifying a point $x\in X(K)$ and a
tangent vector in $v\in T_x(X)$, with $f(\varphi) \in \cL_1(\bA^1)$, for a given $f: X \to \bA^1$, 
correspondingly specifying a point $f(x)$ and a tangent vector $df_x(v)\in T_{f(x)}\bA^1$. 
We can identify group homomorphisms $\chi_1: \bF_q[u]/u^2 \to \bC^*$ with pairs
of characters $\chi,\chi': \bF_q \to \bC^*$. 
Since $T(X)$ is locally trivial, given a choice of a characters $\chi,\chi': K \to \bC^*$, 
the $h_{n,\underline{\varphi}}$ are locally of the form $\chi(f^{(n)}(\underline{x})) 
\chi'(df^{(n)}(\underline{v})) t^n$, for $\underline{\varphi}\in S^n(\cL_1(X))$ corresponding 
to $\underline{x}\in S^n(X)$ with a tangent vector $\underline{v}$. 
The vanishing of the sum follows from the fact that the class $[ \cL_1(X), f_1 ]$ can be
written in $KExp^c(\cV)_K$ as a sum of classes 
\[
[ X_i \times V_i, f\circ \pi_{X_i} + \langle df\circ \pi_{X_i}, \pi_{V_i} \rangle ], 
\] 
where the second term is th linear form $\langle df_x, \pi_{V_i} \rangle : V_i \to K$
for all $x\in X(K)$. For any non-trivial linear form $\lambda : V \to K$, for a finite 
dimensional $K$-vector space $V$, and a character $\chi: K \to \bC^*$, one has
$$ \sum_{v\in V} \chi(\lambda(v)) =0. $$
This shows the vanishing of $\sum_{\underline{v}} \chi'(df^{(n)}_x(\underline{v}))$,
hence of the sum of the $h_{n,\underline{\varphi}}$.
\end{proof}

\vspace{3pt}

The argument above can also be rephrased as showing that all the classes
$[ S^n(\cL_1(X), f_1) ]$ are in the ideal of $KExp^c(\cV)_K$ generated 
by $[\bA^1,id]$, hence trivial in the Grothendieck ring with exponentials $KExp(\cV)_K$
(compare the analogous argument uses in Lemma~1.1.11 and Theorem~1.2.9
of \cite{ChamLoe15}).

\medskip

\begin{proposition}\label{P:6.7.2}

Consider the distribution $\rho=(P_{n,\underline{x}})$ with $n\in \N$ and $\underline{x}\in S^n(X)$ given by
\[
 P_{n,\underline{x}} = \frac{\chi(f^{(n)}(\underline{x}) t^n}{\zeta_\chi((X,f),t)}
 \]
and an increment $h_{n,\underline{x}}(\underline{v})$ with $\underline{v}\in T_{\underline{x}} S^n(X)$
given by $h_{n,\underline{x}}(\underline{v}) = h_{n,\underline{\varphi}}$ as in Proposition~\ref{P:6.1.3}, with
$\underline{\varphi}\in S^n(\cL_1(X))$ determined by $\underline{x}$ and $\underline{v}$. One obtains
a Fisher--Rao tensor $g=(g_{v,w})$ as in Lemma~\ref{L:6.6.1} given by
\[
 g_{v,w} = \frac{\zeta_\chi((X,f),t)}{2} \sum_{n,\underline{x}}  \chi(f^{(n)}(\underline{x})) \, 
\chi'(df^{(n)}_{\underline{x}}(\underline{v})) \, \chi'(df^{(n)}_{\underline{x}}(\underline{w})) t^n. 
\]
\end{proposition}

\vspace{3pt}

\begin{proof}
As in Lemma~\ref{L:6.6.1} above,  the leading order term of
$KL(\rho + h || \rho)$, which defines the Fisher--Rao tensor, is given by 
$\frac{1}{2} Tr(h \rho^{-1} h)$. This gives the Fisher--Rao tensor

\[
g_{v,w}:= \frac{1}{2} \sum_{n,\underline{x}} P_{n,\underline{x}}^{-1} \, \, h_{n,\underline{\varphi_v}}
h_{n,\underline{\varphi_w}},
\]
where we write $\varphi_v, \varphi_w$ for the elements in $\cL_1(X)$  corresponding to
infinitesimal arcs at the point $x$ with tangent directions $v,w\in T_x(X)$, respectively.
More explicitly, on a $X_i \subset X$ that trivializes the tangent bundle, 
\[\begin{aligned}
g_{v,w} &= \frac{\zeta_\chi((X,f),t)}{2}  \sum_{n,\underline{x}} \chi^{-1}(f^{(n)}(\underline{x}))\, t^{-n} \,\, \chi^2(f^{(n)}(\underline{x}))\, t^{2n} \, \chi'(df^{(n)}_{\underline{x}}(v)) \, \chi'(df^{(n)}_{\underline{x}}(w)) \\
&= \frac{\zeta_\chi((X,f),t)}{2} \sum_{n,\underline{x}} \chi(f^{(n)}(\underline{x})) \, \chi'(df^{(n)}_{\underline{x}}(v)) \, \chi'(df^{(n)}_{\underline{x}}(w))\, t^n, 
 \end{aligned}\]
as stated. 
The dependence of $g_{v,w}$ on the pair of characters $\chi,\chi'$ is through the
identification of homomorphisms $\bF_q[u]/u^2\to \C^*$ with pairs $(\chi,\chi')$ discussed in the previous Lemma~\ref{L:6.7.1}, 
through the expression of the increments $h_{n,\underline{\varphi_v}}$ and $h_{n,\underline{\varphi_w}}$ through
$\chi'(df^{(n)}_{\underline{x}}(\underline{v}))$ and $\chi'(df^{(n)}_{\underline{x}}(\underline{w}))$.
\end{proof}

\bigskip

\vspace{3pt}

The Fisher--Rao tensor obtained in this way is not a Riemannian metric in the usual
sense, since it is complex valued and does not define a positive definite quadratic form.
However, it retains some of the significance of the information metric, for instance in
the sense that we can interpret distributions with the same Fisher--Rao tensor in terms
of sufficient statistics. 

\medskip

\subsection{Amari-Chentsov tensor}\label{S:6.8}

In information theory it is customary to use methods from Riemannian geometry, by
expressing the distance between two probability distributions on a sample space in
terms of the Fisher--Rao information metric (\cite{AmNag07}). 

\medskip

\begin{definition} \label{D:6.8.1}

A {\it statistical manifold} is a datum $(M,g,A)$
of a manifold together with a metric tensor and a totally symmetric $3$-tensor $A$,
the Amari-Chentsov tensor,
\[ A_{abc}= A(\partial_a, \partial_b, \partial_c)=\langle \nabla_a \partial_b - \nabla^*_a \partial_b, \partial_c\rangle,\]
where $\nabla$ is a torsion free affine connection such that $\nabla g$ is symmetric and  $\nabla^*$ is its dual connection. 
\end{definition}

\vspace{3pt}

The following classes of statistical manifolds are especially interesting (\cite{AmCi10}, \cite{Ni20}).

\medskip

\begin{definition} \label{D:6.8.2}

A divergence function on a manifold $M$ is a differentiable, non--negative real valued function 
$D(x||y)$, for $x,y\in M$,
that vanishes only when $x=y$ and such that the Hessian in the $x$--coordinates evaluated at $y=x$ is positive definite.
\end{definition}

\vspace{3pt}

A divergence function determines a statistical manifold structure by setting
\[
 g_{ab}=\partial_{x_a} \partial_{x_b} D(x||y) |_{y=x} ,
 \]
\[
A_{abc} = (\partial_{x_a}\partial_{x_b} \partial_{y_c}-\partial_{x_c} \partial_{y_a} \partial_{y_b}) D(x||y) |_{y=x}. 
\]
The Amari-Chentsov tensor $A_{abc}$ obtained in this way vanishes identically if the divergence $D(x||y)$ is
symmetric.

\medskip

\begin{definition} \label{D:6.8.3}

A statistical manifold is induced by a Bregman generator, if it is determined by a divergence function
and there is a potential $\Phi$ such that the divergence function is a Bregman divergence, namely it
satisfies locally
\[
D(x||y) = \Phi(x)-\Phi(y) - \langle \nabla \Phi(y), x-y \rangle. 
\]
\end{definition}

\vspace{3pt}

In particular, one can consider the statistical manifold structure induced by the Shannon entropy
on the space of probability distributions on a set, with the Hessian of the Kullback--Leibler divergence
given by the Fisher-Rao metric,
\[\begin{aligned}
g_{ab}&= \sum_n P_n\, \, \partial_a \log P_n\,\, \partial_b \log P_n\\
&=\sum_n \frac{\partial_a P_n\, \partial_b P_n}{P_n} \\
& = -\sum_n P_n\,\, \partial_a \partial_b \log P_n \\
&=  \partial_a \partial_b KL(P||Q) |_{P=Q} ,
\end{aligned}\]
and the Amari-Chentsov $3$-tensor given by 
\[\begin{aligned} A_{abc} &= \sum_n P_n \, \partial_a \log P_n\,\, \partial_b \log P_n\,\, \partial_c \log P_n\\ 
&=\sum_n \frac{\partial_a P_n\, \partial_b P_n\, \partial_c P_n}{P_n^2} \\
&= (\partial_a\partial_b\partial_{c'}-\partial_c \partial_{a'} \partial_{b'}) KL(P||Q) |_{P=Q},
\end{aligned}\]
where we write $a,b,c$ for the variation indices for $P$ and $a',b',c'$ for $Q$.
This determines an associated flat connection. 
This connection and the one obtained from it by Legendre transform of the potential determine 
the mixing and exponential geodesics in information geometry, see Chapter~2 of \cite{AmNag07}. 

\vspace{3pt}

\begin{lemma} \label{L:6.8.4}

The statistical manifold structure $(M,g,A)$ associated to the Fisher--Rao metric
and Amari--Chentsov tensor obtained from the Kullback--Leibler divergence is
induced by a Bregman generator.
\end{lemma}

\vspace{3pt}

\begin{proof}
One needs to check that the Kullback--Leibler divergence is a Bregman
divergence,
\[
 KL(P||Q) = \Phi(P)-\Phi(Q) - \langle \nabla \Phi(Q), P-Q \rangle. 
 \]
This is the case for the potential $\Phi(P)=-H(P)=\sum_n P_n \log P_n$, 
the negative of the Shannon entropy.  Indeed, we have $\nabla H(Q)=(1+\log Q_n)_n$ and 
\[\begin{aligned}
 \Phi(P)-\Phi(Q) - &\langle \nabla \Phi(Q), P-Q \rangle \\&=\sum_n P_n \log P_n - \sum_n Q_n \log Q_n 
- \sum_n (1+\log Q_n) (P_n -Q_n)\\
& = \sum_n P_n \log P_n -\sum_n P_n \log Q_n =KL(P||Q).
\end{aligned}\]
\end{proof}

\vspace{3pt}

\begin{lemma} \label{L/6.8.5}

The Fisher-Rao metric is related to the Hessian of the Bregman potential by
\[
 \partial_a \partial_b \Phi = g_{ab}+ \langle \nabla \Phi, \partial_a \partial_b \Phi \rangle.
 \]
The Amari-Chentsov tensor is related to the symmetric tensor of the third derivatives of
the Bregman potential by
\[\begin{aligned}
 \partial_a \partial_b \partial_c \Phi(P) =& - A_{abc} +\sum_n (\partial_a \partial_b \partial_c P_n) \log P_n \\
& +\sum_n\frac{\partial_a \partial_b P_n\, \partial_c P_n + \partial_b\partial_c P_n\, \partial_a P_n  + \partial_a \partial_c P_n \, \partial_b P_n}{P_n}.  
\end{aligned}\]
\end{lemma}

\vspace{3pt}

\begin{proof}
One has
\[\begin{aligned}
 \partial_a \partial_b \Phi(P)&=\sum_n (\partial_a\partial_b P_n) \log P_n 
+\sum_n \frac{\partial_a P_n\, \partial_b P_n}{P_n}\\
& =\langle \nabla \Phi, \partial_a \partial_b \Phi \rangle + g_{ab},
\end{aligned}\]
where we used the fact that $\sum_n \partial_a P_n = \sum_n \partial_a \partial_b P_n =0$. The identity for
the Amari-Chentsov tensor follows immediately by applying $\partial_c$ to the previous identity.
\end{proof}

\vspace{3pt}

In the case of the Kullback--Leibler divergence for zeta functions on the Grothendieck ring with
exponentials introduced in the previous subsections, we do not impose the positivity of the divergence
and the positive definiteness of the Hessian, since these take complex values, but we still have
an associated Fisher--Rao $2$--tensor (Proposition~\ref{P:6.7.2} and an Amari--Chentsov $3$--tensor given as follows.

\vspace{3pt}

\begin{lemma} \label{L:6.8.6}

 The Amari--Chentsov $3$--tensor associated to the distribution 
$P_{n,\underline{x}}$
and the increments $h_{n,\underline{x}}(v)$ as in Proposition~6.7.2, is given by
\[
A_{u,v,w} = \zeta((X,f),t) \sum_{n,\underline{x}} \chi(f^{(n)}(\underline{x}) \,
\chi'(df^{(n)}_{\underline{x}}(u))\chi'(df^{(n)}_{\underline{x}}(v)) \chi'(df^{(n)}_{\underline{x}}(w))\, t^n. 
\]
\end{lemma}

\begin{proof}
Computing as in Proposition~\ref{P:6.7.2} we obtain 
{\Small \[\begin{aligned}
A_{u,v,w} &=\sum_{n,\underline{x}} P_{n,\underline{x}}^{-2}\,\, h_{n,\underline{x}}(u) h_{n,\underline{x}}(v) h_{n,\underline{x}}(w)\\
& = \zeta((X,f),t)^2 \sum_{n,\underline{x}} \chi^{-2}(f^{(n)}(\underline{x}))  t^{-2n} \chi^3(g^{(n)}(\underline{x}) t^{3n} 
\chi'(df^{(n)}_{\underline{x}}(u))\chi'(df^{(n)}_{\underline{x}}(v)) \chi'(df^{(n)}_{\underline{x}}(w)). 
\end{aligned}\]}
\end{proof}

\vspace{3pt}

\subsection{Amari--Chentsov tensor and Frobenius manifolds}\label{S:6.9}

As we have recalled in the previous sections, 
a {\it Frobenius manifold} is a datum $(M,g,\Phi)$ of a manifold $M$ 
with a flat metric $g$ with local flat coordinates $\{ x^a \}$, and a potential $\Phi$
such that the tensor $A_{abc}=\partial_a \partial_b \partial_c \Phi$ defines
an {\it associative} multiplication 
\[ 
\partial_a \circ \partial_b = \sum_c {A_{ab}}^c \partial_c , 
\]
with the raising and lowering of indices done through the metric tensor,
${A_{ab}}^c= \sum_e A_{abe} e^{ec}$. The associativity condition
for the multiplication is expressed as the WDVV nonlinear differential
equations in the potential $\Phi$. In terms of the $3$--tensor $A_{abc}$ the
associativity is simply expressed as the identity
\[ 
A_{bce} g^{ef} A_{fad} = A_{bae} g^{ef} A_{fcd}, 
\]
where $g^{ab}$ is the inverse of the metric tensor.
One usually assumes that the manifold $M$ is complex, but we will
not necessarily require it here.

\vspace{3pt}

The {\it first structure connection} on a Frobenius manifold is given
in the flat coordinate system by
\[
 \nabla_{\lambda,\partial_a} \partial_b = \lambda \sum_c {A_{ab}}^c \partial_c
=\lambda \partial_a \circ \partial_b, 
\]
for $\lambda$ a complex parameter. The associativity of the product and 
the existence of a potential function are equivalent to the connection 
$\nabla_\lambda$ being flat. 

\vspace{3PT}

Given a statistical manifold $(M,g,A)$, as in the previous subsection, with
$g_{ab}$ the matrix of the Fisher-Rao metric tensor, with inverse $g^{ab}$ 
given by the covariance matrix, and $A_{abc}$ the Amari--Chentsov tensor.
It is natural to ask whether this structure also gives rise to a Frobenius
manifold structure, namely whether the Amari--Chentsov tensor defines
an associative multiplication in the tangent bundle of $M$. This is equivalent
to the condition that Amari--Chentsov tensor satisfies the associativity identity
\[
A_{bce} g^{ef} A_{fad} = A_{bae} g^{ef} A_{fcd}. 
\]
In the case of a $3$--tensor $A_{abc}=\partial_a \partial_b \partial_c \Phi$, for
a potential $\Phi$, this equation becomes the WDVV equation for $\Phi$
of Frobenius manifold theory. However, in the case of statistical manifolds,
usually the Amari--Chentsov tensor differs from the tensor of third derivatives
of the potential as discussed above. The condition that replaces the WDVV equation
is then of the following form.

\vspace{3pt}

\begin{proposition}\label{P:6.9.1}

 Let $(M,g,A)$ be a statistical manifold that 
induced by a Bregman generator, as in Definition~\ref{D:6.8.3}. Then the
Amari-Chenstov tensor defines an associative multiplication on the
tangent space $TM$,
\[
 \partial_a \circ \partial_b =\sum_c A_{ab}^c \partial_c, 
 \]
iff the Bregman potential $\Phi$ satisfies the identity

\Small{
\[\begin{aligned}
\langle  &\partial_e \nabla\Phi(P),\partial_a \partial_b P\rangle g^{ef} \langle \partial_f \nabla\Phi(P), \partial_c\partial_d P\rangle + 
 \langle \partial_e \nabla\Phi(P),\partial_a \partial_b P\rangle g^{ef} \langle \partial_c\partial_d \nabla\Phi(P), \partial_f P \rangle  \\
&+\langle \partial_a \partial_b \nabla\Phi(P), \partial_e P \rangle g^{ef} \langle \partial_f \nabla\Phi(P), \partial_c\partial_d P\rangle + 
 \langle \partial_a \partial_b \nabla\Phi(P), \partial_e P \rangle g^{ef} \langle \partial_c\partial_d \nabla\Phi(P), \partial_f P \rangle   \\
= & \langle \partial_e \nabla\Phi(P),\partial_a \partial_c P\rangle g^{ef} \langle \partial_f \nabla\Phi(P), \partial_b\partial_d P\rangle + 
\langle \partial_e \nabla\Phi(P),\partial_a \partial_c P\rangle g^{ef} \langle \partial_b\partial_d \nabla\Phi(P), \partial_f P \rangle \\
&+ \langle \partial_a \partial_c \nabla\Phi(P), \partial_e P \rangle g^{ef} \langle \partial_f \nabla\Phi(P), \partial_b\partial_d P\rangle + 
 \langle \partial_a \partial_c \nabla\Phi(P), \partial_e P \rangle g^{ef} \langle \partial_b\partial_d \nabla\Phi(P), \partial_f P \rangle . 
 \end{aligned}\]
 }
\end{proposition}

\begin{proof}

 This follows by a direct computation from the dependence of the divergence function on the
Bregman potential,
\[ 
 D(P||Q) = \Phi(P)-\Phi(Q) -\langle \nabla\Phi(Q),Q-P\rangle, 
 \]
and the expression of the Amari-Chenstov tensor as a function of the divergence function
\[
 A_{abc}= (\partial_a\partial_b\partial_{c'}-\partial_c \partial_{a'} \partial_{b'}) D(P||Q) |_{P=Q}, 
 \]
which gives
\[
 A_{abc}=\langle \partial_c \nabla \Phi(P), \partial_a\partial_b P\rangle + \langle \partial_a\partial_b
\nabla \Phi(P), \partial_c P \rangle. 
\]
Replacing these expressions in the associativity identity gives the condition stated above.
\end{proof}

\smallskip

We also obtain the following special case from the previous proposition together with Lemma~6.8.5.

\vspace{3pt}

\begin{corollary}\label{C:6.9.2}

 If the dependence of $P$ on the deformation parameters is linear, so that 
$\partial_a \partial_b P=0$, then the associativity condition reduces to
\small{\[
 \langle \partial_a \partial_b \nabla\Phi(P), \partial_e P \rangle g^{ef} \langle \partial_c\partial_d \nabla\Phi(P), \partial_f P \rangle =  \langle \partial_a \partial_c \nabla\Phi(P), \partial_e P \rangle g^{ef} \langle \partial_b\partial_d \nabla\Phi(P), \partial_f P \rangle. 
 \]}
In this case the Amari-Chenstov tensor is the tensor of third derivatives of the Bregman potential, so the
equation above is an equivalent formulation of the WDVV equation for the potential.
\end{corollary}

\medskip

In the previous sections we have discussed some instances of $F$-manifold structures in information geometry.
We have also discussed here above the WDVV equation in the case of divergence functions with Bregman potential.
In general, the difficulty of upgrading from $F$-manifolds to Frobenius manifolds lies in the flatness of the metric.
In the context of information manifolds, flat structures typically arise when one considers $1$-parameter $\alpha$-deformations
of a statistical manifold $(M,g,A)$ with one-parameter dual families of connections
$$ \Gamma^\alpha_{ijk}=\Gamma^{LC}_{ijk} - \frac{\alpha}{2} A_{ijk}, \ \ \ \ \
\Gamma^{-\alpha}_{ijk}=\Gamma^{LC}_{ijk} + \frac{\alpha}{2} A_{ijk} $$
for $\Gamma^{LC}$ the Levi-Civita connection, with flatness achieved at special values $\alpha=\pm 1$.
The relevance of these flat structures for $F$-manifolds and Frobenius manifolds constructions is discussed in
\cite{CCN21}.

\smallskip

Regarding the WDVV equation, in the case we discussed in Section~\ref{S:4}, with characteristic functions $\varphi_V$ on 
convex cones $V$ and the metric $g_{ij} =\partial_i \partial_j \log \varphi_V$, there is also a direct elementary way to see
that the WDVV equation is satisfied. Indeed, consider for simplicity the case of the orthant $V=\R_+^n$, for which we
can write
$$ \varphi_V(X)=\int e^{-\langle X, Y \rangle} dv(Y)=\prod_{i=1}^n \int e^{-X_i Y_i} dY_i =\prod_{i=1}^n \varphi_i(X_i) $$
with $\varphi_i(X_i) =\int e^{-X_i Y_i} dY_i$. For a subset $I\subset \{ 1, \ldots, N \}$ we write
$$ \varphi_I :=\prod_{i\in I} \varphi_i  $$
and $\varphi=\varphi_V$.
We also use the notation
$$ \psi_i :=\int Y_i e^{-X_i Y_i} dY_i =- \partial_i \phi_i, \ \ \  \psi_I :=\prod_{i\in I}\psi_i $$
$$ \psi_{i,k_i}: =\int Y^{k_i}_i e^{-X_i Y_i} dY_i, \ \ \  \psi_{I,\underline{k}}:=\prod_{i\in I}\psi_{i,k_i} .$$ 
We then have
$$ \partial_i \log \varphi =\frac{ \partial_i \varphi}{\varphi}=\frac{-\psi_i \, \varphi_{\{i \}^c}}{\varphi} =-\frac{\psi_i}{\varphi_i} $$
so that the metric tensor takes the simple diagonal form 
$$ g_{ij}=\partial_i \partial_j \log \varphi = \delta_{ij} \partial_i  \frac{-\psi_i}{\varphi_i} =\delta_{ij} (\frac{\psi_{i,2}}{\varphi_i} - \frac{\psi_i^2}{\varphi_i^2}), $$
where positivity follows from the Cauchy-Schwartz inequality
$$ (\int Y^2 e^{-XY} dY)(\int e^{-XY} dY)\geq (\int Y e^{-XY} dY)^2 . $$
This means that the $3$-tensor $A$ is also diagonal with
$$ A_{iii}= -\frac{\psi_{i,3}}{\varphi_i}  +3 \frac{\psi_i \psi_{i,2}}{\varphi_i^2}-2\frac{\psi_i^3}{\varphi_i^3}. $$
Thus, the WDVV equation 
$$ A_{bce} g^{ef} A_{fad} = A_{bae} g^{ef} A_{fcd} $$
is then trivially satisfied, as both sides are equal to $A_{iii}^2 g^{ii}$. This gives an $F$-structure, 
though not a Frobenius structure due to lack of flatness.

\smallskip

In the Frobenius manifold structures that arise in the theory of Gromov--Witten invariants, the potential $\Phi$ that
satisfies the WDVV equations is not in itself the metric potential, unlike in the example above. Indeed, it arises through
a {\em deformation} of the metric structure 
$$ g=\frac{1}{2} \sum_{a,b} \eta_{ab} dt^a dt^b $$
given by the usual intersection product 
$$ \eta_{ab} =\int \gamma_a \cup \gamma_b $$
on the cohomology $H^*(X,\C)$ (with $\gamma_a$ a homogeneous basis of cycles and $t^a$ dual basis).
It is the deformation of cohomology to quantum cohomology that gives rise to the Frobenius manifold potential 
$$ \Phi=\sum_{n\geq 3} \frac{1}{n!} \sum_{a_1,\ldots, a_n} t^{a_1} \cdots t^{a_n}\, I^X_{g,n}(\gamma_{a_1},\ldots, \gamma_{a_n})\, , $$
where the coefficients $I^X_{g,n}(\gamma_{a_1},\ldots, \gamma_{a_n})$ are the Gromov-Witten invariants 
that count holomorphic curves of genus $g$ with homological constraints imposed at $n$ points of the curve,
instead of counting intersection numbers as in ordinary cohomology. 

\smallskip

This suggests that in the context of information geometry one should expect a deformation of the type of
$F$-structure considered above, with a new potential that incorporates higher mutual information structures,
of the type known to have good cohomological interpretations (see \cite{BaBe15}, \cite{Vig17}).
A possible form of statistical Gromov--Witten invariants was proposed in \cite{CCN21}. 
Regarding the difference between $F$-manifold and Frobenius manifold, it seems that the most suitable
level of structure for the information geometry case will fall in between, and is best described by the
notion of {\em $F$-manifold with flat structure} introduced by the second author in \cite{Ma05}.

\vspace{5pt}

\end{document}